\RequirePackage{fix-cm}
\documentclass[final]{svjour3}                     % onecolumn (standard format)
\smartqed  % flush right qed marks, e.g. at end of proof
\usepackage{subfig}
\usepackage{amsmath}
\usepackage{amssymb}
\usepackage{setspace}
\usepackage{algorithmic}
\usepackage{algorithm}
\usepackage{tikz-cd}
\usepackage{tikz}
\usepackage[pagebackref=true,breaklinks=true,colorlinks,bookmarks=false]{hyperref}

\begin{document}

\title{Inconsistent Surface Registration via Optimization of Mapping Distortions}

\titlerunning{Inconsistent Surface Registration}        % if too long for running head

\author{Di Qiu and Lok Ming Lui}

%First Author         \and
%        Second Author %etc.
%}


\institute{Qiu Di and Lok Ming Lui \at
              Department of Mathematics, The Chinese University of Hong Kong, Shatin, New Territories, Hong Kong  \\
%              Tel.: +123-45-678910\\
%              Fax: +123-45-678910\\
%              \email{lmlui@example.com}           %  \\
}
\date{Received: date / Accepted: date}

\maketitle

\begin{abstract}
We address the problem of registering two surfaces, of which a natural bijection between them does not exist. More precisely, only a partial subset of the source surface is assumed to be in correspondence with a subset of the target surface. We call such a problem an {\it inconsistent surface registration (ISR)} problem. This problem is challenging as the corresponding regions on each surface and a meaningful bijection between them have to be simultaneously determined. In this paper, we propose a variational model to solve the ISR problem by minimizing mapping distortions. Mapping distortions are described by the Beltrami coefficient as well as the differential of the mapping. Registration is then guided by feature landmarks and/or intensities, such as curvatures, defined on each surface. The key idea of the approach is to control angle and scale distortions via quasiconformal theory as well as minimizing landmark and/or intensity mismatch. A splitting method is proposed to iteratively search for the optimal corresponding regions as well as the optimal bijection between them. Bijectivity of the mapping is easily enforced by a thresholding of the Beltrami coefficient. We test the proposed method on both synthetic and real examples. Experimental results demonstrate the efficacy of our proposed model.
\end{abstract}

\section{Introduction}
Surface registration aims to find meaningful point-wise correspondence between two surfaces embedded in $
\mathbb{R}^3$. It has important applications in various fields, such as in computer graphics, computer vision and medical imaging. For instances, the surface registration problem in computer vision and graphics aims to find point-wise correspondence in order to perform shape analysis, relational learning, to transfer motions, textures between shapes; in medical imaging, it is necessary to find one-to-one point-wise correspondence between the target and the template anatomical surfaces so that the data defined on the surfaces can be compared meaningfully. Very often, a desirable registration map should be much more complex than a global rigid or affine motion.

Due to its importance, various registration models have been proposed. Existing approaches usually assume a global bijection between the two surfaces to be registered if the surfaces are closed without boundary. For the registration between open domains, prescribed boundary condition is imposed. In a practical situation, a natural bijection between two shapes may not exist. Usually, only a subset $\Omega_1$ of the source surface $S_1$ is in correspondence with a subset $\Omega_2$ of the target surface $S_2$. This problem arises in many real situations, such as surface stitching, surface matching of incomplete anatomical structures and so on. We refer to this kind of registration problem as the {\it inconsistent surface registration (ISR)} problem, and $S_1, S_2$ to be an {\it inconsistent pair} of surfaces. Importantly, we assume the registration map from our ISR problems satisfy the {\it maximality property}: 
\begin{equation} \label{eq:maximal}
f(\Omega_1) = \Omega_2 = f(S_1) \cap S_2,
\end{equation}
where the registration $f:\Omega_1 \to \Omega_2$ is extended to a deformation $f:S_1 \to \mathbb{R}^3$ with abusing the notation.
This property means that our corresponding regions can be obtained from deforming the source surface and taking the intersection with the target surface. Thus it excludes the scenarios where the correspondence can only be defined in multiple, mutually disconnected regions. 

To solve the problem, the corresponding regions on each surface as well as a meaningful bijection between them have to be simultaneously found. Mathematically, this problem can be formulated as follows. Given two surfaces $S_1$ and $S_2$ to be registered, we look for optimal subsets $\Omega_1^* \subset S_1$ and $\Omega_2^* \subset S_2$, as well as an optimal registration $f$, such that when restricted to $\Omega_1^*$, $f:\Omega_1^*\to \Omega_2^*$ is a bijection that satisfies the prescribed mapping constraints. The mapping constraints are often given by feature landmarks and intensities, such as surface curvatures, defined on each surface. The ISR problem can then be described as the following optimization problem:
\begin{equation}\label{optimizationISR}
	(\Omega_1^*, f^*) = {\bf argmin}_{\Omega_1, f:\Omega_1\to \Omega_2} \{E_{fid}(\Omega_1,f) + E_{reg}(f)\}
\end{equation}
\noindent where $E_{fid}$ denotes the data fidelity energy, which is usually an $L^2$ loss on the intensity differences; $E_{reg}$ denotes the regularization term, which enforces the mapping $f$ to be smooth and bijective. Further constraints, such as landmark constraints, may be imposed on the solution $f$. Note that this problem is different from the conventional registration problem, since apart from the mapping problem, one also needs to find the optimal region $\Omega_1^*$. 

In this paper, we propose a variational model to solve the above ISR problem through minimizing mapping distortions. We capture the mapping distortions by a geometric quantity, called the {\it Beltrami coefficient}, from quasi-conformal theory, together with the singular values of the differential of the mapping. The main idea of the proposed model is to control the angle and scale distortions, allow the matching domain to evolve in a way that minimizes the landmark and intensity mismatching. A splitting method is proposed to iteratively find the optimal corresponding regions on each surface and the optimal bijection between them. The incorporation of Beltrami coefficient in our model allows us to conveniently enforce the local bijectivity and smoothness of the mapping. More specifically, the local bijectivity of the mapping can be easily enforced by a thresholding of the Beltrami coefficient, and the smoothness of the mapping can be translated to the smoothness of the coefficient. Numerous experiments have been carried out on both synthetic and real data, which demonstrate the effectiveness of the proposed model to solve the ISR problem.

The contributions of this paper are three-fold.
\begin{enumerate}
	\item Firstly, we propose to formulate the ISR problem as an optimization problem over the spaces of sub-regions on each surface and bijective mappings between them. The regularization of the mapping is based on the differential and the Beltrami coefficient of the mapping from quasiconformal theory, with which the bijectivity and smoothness of the mapping can be easily enforced. 
	\item Secondly, we propose an algorithm to obtain a free boundary deformation with controlled scale distortions. This algorithm is useful for solving the optimization problem for ISR in our formulation.
	\item Thirdly, a splitting method is proposed to solve the ISR optimization problem, which iteratively searches for the optimal corresponding regions as well as the optimal bijection between them. 
\end{enumerate}

The rest of the paper is organized as follows. In Section \ref{sec:2}, relevant previous works are reviewed. Some basic mathematical background are explained in Section \ref{sec:3}. Our proposed model is discussed in details in Section \ref{sec:4}. Details of the main algorithm to solve the proposed registration model are explained in Section \ref{sec:5}. Experimental results are shown in Section \ref{sec:5}. The paper is concluded in Section \ref{sec:7}.

\section{Related works} \label{sec:2}
Surface registration in a non-rigid, deformable setting has been an active and challenging area of research. Since the literature is vast, below we will mention works that are related to ours in terms of the problem setting as well as the techniques deployed, but definitely they will not consist a complete survey.

We approach the ISR problem in 2D parametrization domain. Explicitly, we use the conformal, also known as the intrinsic parametrizations \cite{pinkall1993computing,levy2002least,desbrun2002intrinsic,gu2003global} of the the surfaces, which faithfully preserve the local geometry of the surfaces. And crucially, we register two inconsistent surfaces via optimized deformation in the 2D parametrization domain. Using conformal parametrization is common in recent works on {\it globally bijective} surface registration possibly with large deformation. For example, for surfaces with prescribed boundary correspondence as in Lam {\it et al.} \cite{lam2014landmark}, for closed surfaces without boundaries as in \cite{choi2015flash,lui2014geometric}. Like ours, these works are based on landmark-intensity matching with geometric regularization, where the regularization term is formulated as the local {\it conformal} or {\it angular distortion} of the deformation map. We refer to  \cite{lam2014landmark} and references therein for further related work on the landmark-intensity hybrid registration. However, the local {\it scale distortion} of the deformation map was not considered by them, mainly because the registration domain are fixed by boundary constraints \cite{lam2014landmark} or in very special forms \cite{choi2015flash,lui2014geometric}, and thus cannot handle ISR problems. 
 
In our work we formulate the scale distortion in a uniform framework with the conformal distortion, namely the geometric properties of the mapping differential in quasiconformal theory \cite{astala2008elliptic}. In this respect, our work is also closely related to the As-Rigid-As-Possible mesh deformation paradigm \cite{sorkine2007rigid,igarashi2005rigid} and its extensions to allow large non-isometric deformations \cite{lipman2012bounded}. Under such a framework geometric smoothness and local bijectivity are non-trivial to achieve and recently many additional optimization methods, {\it e.g.} \cite{kovalsky2015large,rabinovich2017scalable} are proposed to solve this problem. However, we will show that with conformal parametrization in 2D, it is very easy to achieve geometric smoothness and local bijectivity using the quasiconformal methods. Indeed, our work gathers both the ideas from landmark-intensity registration and mesh deformation modelling to solve the ISR problem possibly with large deformation.

In case the ground truth registration map is a global rigid or affine motion, then the ISR problem can also be posed as a fusion problem between surfaces and can be solved by iterative closest point type of methods, please see \cite{rusinkiewicz2001efficient,li2008global,amberg2007optimal,bronstein2006generalized} and the references therein. Directly applying these methods to our ISR problem will not work, since these methods either fail to handle large deformation or the resulting deformation of the source surface only approximate the target surface, and thus does not yield and mapping between them. We illustrated these issues in Example 8 of Section 7.  On the other hand, the feature matching approach is often used in case the matching is allowed to be sparse or possible locally non-injective. In this regard, methods based on the functional correspondence framework \cite{rodola2017partial,ovsjanikov2016computing} have also been proposed. Differently, we note that our goal is to obtain a locally bijective mapping possibly with large deformation between sub-regions of the surfaces to be registered.  

\section{Preliminaries on quasiconformal deformations} \label{sec:3}
In this section, we introduce the basic mathematical concepts related to quasiconformal deformations. Since our discussion will be about local differential properties of the mapping, it suffices to consider planar domains in $\mathbb{R}^2 \cong \mathbb{C}$ of disk topology, equipped with the Euclidean metric $e$. 

Consider a diffeomorphism $f:X\to \mathbb{C}$ for a domain $X \subset \mathbb{C}$.  This mapping $f$ induces a pullback metric $f^*e$ on $X$, and this metric on $X$ reveals geometric properties of $f$.
Let us write $f^*e$ as a matrix field $H:X\to{\bf S}_{++}$ defined on $X$, where $\mathbf{S}_{++}$ denotes the space of symmetric positive definite matrices. Then $H$ and the differential $Df$ of the mapping $f$ satisfy the following nonlinear equation:
\begin{equation} \label{eq:nonlinear}
Df(z)^{T}Df(z)=H(z), \; z\in X
\end{equation}
For each $z\in X$, we further factorize $H(z)$ as a product of $\det(Df)(z)$ and a positive definite matrix {\it with unit determinant} $Q(z)$
\[
H(z)=\det(Df)(z)\cdot Q(z),
\]
Here, we have used the fact that $\det(Df)(z)>0$ since $f$ is assumed to be a diffeomorphism, and thus orientation preserving. This factorization signals a nice ``linearization" of the nonlinear equation \eqref{eq:nonlinear}, by multiplying on the left the inverse of $Df(z)^{T}$ on both sides of \eqref{eq:nonlinear}. After some algebraic operations, this will lead to the Beltrami equation
\begin{equation} \label{eq:beltrami}
\frac{\partial}{\partial\bar{z}}f(z)=\mu(z)\frac{\partial}{\partial z}f(z),
\end{equation}
where  $\mu=\frac{q_{11}-q_{22}+2iq_{12}}{q_{11}+q_{22}+2}$, $Q = (q_{ij})_{1\leq i,j\leq 2}$ and we use the complex notation $\frac{\partial f}{\partial\bar{z}}=(u_{x}-v_{y})/2+i(u_{y}+v_{x})/2$, $\frac{\partial f}{\partial z}=(u_{x}+v_{y})/2+i(-u_{y}+v_{x})/2$. It is easy to check that the local diffeomorphism condition is equivalent to $|\mu|<1$ \cite{astala2008elliptic}.
We call injective solutions to the Beltrami equation with Beltrami coefficient $\mu$ quasiconformal mappings.
A special case when $\mu=0$ on $X$, the equation becomes $\partial_{\bar{z}}f=0$ and is the well-known Cauchy-Riemann equation. In this case, the mapping $f$ is called a conformal mapping. Geometrically, the Beltrami coefficient encodes the infinitesimal angle distortion of the mapping. This can be visualized by the fact that quasiconformal mapping maps infinitesimal circles to ellipses. See Figure \ref{fig:qc1}.

\begin{figure}[t]
    \centering
    \includegraphics[height=4.5cm]{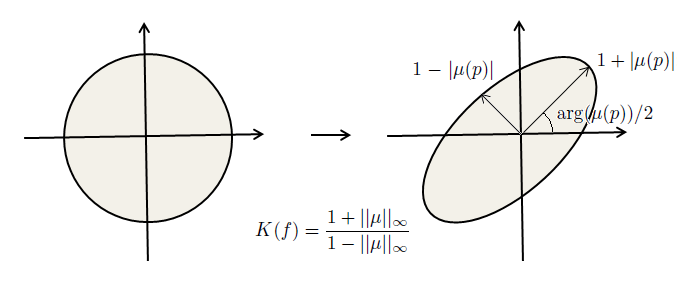}
    \caption{A quasiconformal mapping maps infinitesimal circles to ellipses. The local geometric distortion under the quasiconformal map can be measured by the Beltrami coefficient.}
    \label{fig:qc1}
\end{figure}
\begin{figure}[t]
    \centering
    \includegraphics[height=4.5cm]{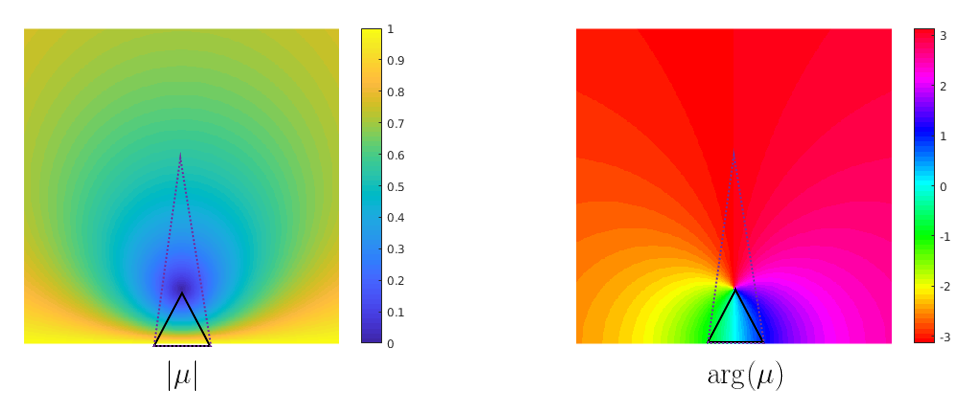}
    \caption{The mapping varies smoothly as the Beltrami coefficient varies smoothly. Here in the discrete case, we map the solid triangle to the dotted triangle, whose two vertices at the base are fixed. The third vertex of the dotted triangle varies in $\mathbb{C}$, which gives rise to various corresponding Beltrami coefficients $\mu = |\mu| e^{i \arg(\mu)}$. We visualize the magnitude and argument of the Beltrami coefficient using a colormap. Best viewed in color.}
    \label{fig:qc2}
\end{figure}

The strength of angle distortion can be measured by the condition number of the mapping differential $Df(z)$ 
\[
K(z)=\frac{1+|\mu(z)|}{1-|\mu(z)|}=\frac{\sigma_{1}(Df(z))}{\sigma_{2}(Df(z))},
\]
where $\sigma_{1}(Df)$ and $\sigma_{2}(Df)$ are the largest and smallest singular values of the mapping differential $Df$, which we call the {\it principal distortions}. One can also observe that as $|\mu|$ tends to $1$, the distortion blows up to infinity. Note that Beltrami coefficients do not control the local area/scale distortions, which are controlled by the Jacobian
\[\det(Df)=\sigma_{1}(Df)\cdot\sigma_{2}(Df),
\]
which is previously factored out from $H$.  We can control the principal distortions by requiring them bounded by some prescribed values $K_1\geq  \sigma_1 \geq \sigma_2 \geq K_2>0$. This means the scales of stretching in each principal direction are bounded by $K_1$ and $K_1$. Thus it can limit the area distortion of the mapping. As a result, the local area distortion as captured by $\text{det}(Df(z))$ is bounded by $K_1^2$ and $K_2^2$.  In particular, the local area is preserved if we set $K_1 = K_2 = 1$.

The angle and area distortions are coupled by a global integrability condition. To be more precise, given Beltrami coefficient $\mu$ with $\|\mu\|_\infty < 1$, the equation has a unique solution up to post-compositions of conformal mappings \cite{astala2008elliptic}. This is known as the following {\it measurable Riemann mapping theorem} \cite{astala2008elliptic}.
\begin{theorem} Let $\mu(z) \leq k < 1$ for every $z\in \mathbb{C}$. Then there is a solution $f$ to \eqref{eq:beltrami} which is a homeomorphism of $\mathbb{C}$. Furthermore, if we require $f(0) = 0, f(1) = 1, f(\infty) = \infty$, then the solution is unique.

\end{theorem}
% Thus, a unique solution can be obtained, in particular, if the boundary value of the mapping is known, or it can be the {\it minimal condition} as in the Riemann mapping theorem, namely fixing both the mapping's value as well as its complex derivative at a point. 

Finally, we have a Schauder estimate \cite{astala2008elliptic} that relates the regularity of the Beltrami coefficients to the regularity of the associated mapping.
\begin{theorem} \label{thm:schauder}
Suppose that $f\in W^{1,2}_{loc}(\mathbb{C}, \mathbb{C})$ is the solution to \eqref{eq:beltrami}, where the Beltrami coefficient $\mu(z) \in  C^{l,\alpha}_{loc}(\mathbb{C}, \mathbb{C})$, $\|\mu\|_{\infty} <1$. Then $f \in C^{l+1,\alpha}_{loc}(\mathbb{C},\mathbb{C})$.
\end{theorem}
In the discrete scenario where the mapping maps a triangle to another, when the Beltrami coefficient varies smoothly, the shape of the mapped triangle also varies smoothly. We visualize the smoothness property in Figure \ref{fig:qc2}.

The above observations about quasiconformal mappings will lead us to a novel free boundary deformation model proposed in this paper, where we take both angle and area distortions into account.

\section{Proposed model} \label{sec:4}

In this section, we describe our proposed model for ISR in details. Our strategy is to using the intensity and landmark matching to guide the deformation to find the optimal bijection as well as finding the optimal corresponding regions on each surface. Usually, apriori intensity information on the surface, such as the curvature, will be provided. In some situations, corresponding feature landmarks on each surface can be delineated. One can formulate the energy as
\begin{equation}\label{Model1}
	E_{ISR}(\Omega_1,f) = E_{fid}(\Omega_1,f) + E_{reg}(f)
\end{equation}
where $\Omega_1 \subset S_1$, $E_{fid}$ is the fidelity term that guides the registration map according to the matching error of intensities on the corresponding regions, $E_{reg}$ is the regularization term for the mapping $f$ that enhances smoothness and reduces local geometric distortions under $f$. The optimization is also subject to the constraint that $f$ is in the {\it admissible set} set $\mathcal{A}$, which constrains the mapping $f$. $\mathcal{A}$ is often defined based on some prescribed requirements according to the problem, such as the landmark constraints, and in our case, bounds on singular values of the mapping differential $Df$. 

The above optimization problem can be simplified by conformally parametrizing $S_1$ and $S_2$ into $\mathbb{C}$. Let $\phi_1:X_1 \to S_1$ and $\phi_2:X_2 \to S_2$ be the global conformal parametrizations of $S_1$ and $S_2$ respectively. Denote $\Omega_2 = f(\Omega_1) $. The ISR problem can be reduced to finding the optimal corresponding regions $\phi_1^{-1}(\Omega_1^*)$ and $ \phi_2^{-1}(\Omega_2^*)$, as well as the optimal bijection $\phi_2^{-1}\circ f^*\circ\phi_1$ when restricted to the corresponding region. This is schematically shown in the following diagram.
\begin{equation}
\begin{tikzcd} \label{diagram}
  S_1 \supset \Omega_1^* \arrow[r,"f^*"]
&   \Omega_2^* \subset S_2 \\
 X_1 \supset \phi_1^{-1}(\Omega_1^*)   \arrow[u, "\phi_1"] \arrow[r,  "\phi_2^{-1} \circ f^* \circ \phi_1"]
& \phi_2^{-1}(\Omega_2^*) \subset  X_2  \arrow[u, "\phi_2"]
\end{tikzcd}
\end{equation}
As such, we will simply discuss the registration problem between two inconsistent 2D domains $X_1\subset \mathbb{C}$ and $X_2\subset \mathbb{C}$, and omit $\phi_1, \phi_2$ from now on. In the following section we will relax the formulation of $E_{ISR}$ to make it tractable to solve in practice.

\subsection{Choices of $E_{fid}$ and $E_{reg}$}
In this subsection, we discuss our choices of fidelity term $E_{fid}$, regularization term $E_{reg}$ and admissible set $\mathcal{A}$.
The fidelity term $E_{fid}(\Omega_1,f)$ guides the registration map $f:\Omega_1 \to \Omega_2$, according to the intensity matching on the corresponding regions $\Omega_1$ and $\Omega_2$, but involving $\Omega_1$ as a variable makes it difficult to solve. We thus consider the following relaxation of the fidelity term:
\begin{equation}\label{fidelity}
	E_{fid}(f) := \int_{f(\Omega_1)} (I_1\circ f^{-1} - I_2)^2 = \int_{f(X_1)\cap X_2} (I_1\circ f^{-1} - I_2)^2
	, 
\end{equation}
where we have extend the registration mapping $f:\Omega_1 \to \Omega_2$ to the deformation mapping $f:X_1 \to \mathbb{C}$  with abusing the notation, $I_1:S_1\to \mathbb{R}$ and $I_2:S_2\to \mathbb{R}$ are intensities defined on $S_1$ and $S_2$ respectively and $I_1\circ f^{-1}$ is the deformed image of $I_1$ under the deformation $f$. Here, the domain of integration $f(\Omega_1)$ is simplified by observing $f(\Omega_1)= \Omega_2 = f(X_1)\cap X_2$, which follows from our assumption of the maximality property \eqref{eq:maximal} of the registration. As such, the fidelity term can be simplified to be dependent on the deformation map $f:X_1 \to \mathbb{C}$ only. We thus aim to minimize the intensity mismatching under $f$ on the region $f(X_1) \cap X_2$.  It serves for two purposes. First, it guides the mapping $f$ by matching the intensities as well as possible, so that a meaningful point-wise correspondence between the two surfaces can be obtained. Second, it optimizes $ f(X_1)\cap X_2$ using the intensity information defined on each surface. Of course, the intensity matching term is equal to $0$ in the trivial case when $f(X_1)\cap X_2$ is an empty set. This is the case when the two surfaces are not in correspondence with each other at all. Thus, the regularization term $E_{reg}$, as well as the landmark constraints, play an important role to avoid this trivial and meaningless case.   Therefore, once the optimal mapping $f^*:X_1 \to \mathbb{C}$ is obtained, the optimal subsets can be obtained by taking intersection $\Omega_2^* = f^*(X_1)\cap X_2$, and inverse mapping $\Omega_1^* = (f^*)^{-1}(\Omega_2^*)$.

The choice of the regularization is crucial in our model. 
From Theorem \ref{thm:schauder} it is natural to consider the following regularization term to find a deformation $f:X_1\to \mathbb{C}$ that minimizes:
\begin{equation} \label{subproblem: 1}
E_{reg}(f) = \frac{1}{2}\int_{X_1} |\nabla \mu(f)|^2,
\end{equation}
where $\mu(f)$ denotes the Beltrami coefficient of $f$. Thus, minimizing this energy will promote the smoothness of the mapping $g$. We call such a process to minimize $E_{reg}$ a {\it geometric smoothing process}.  
Hence our initial version of ISR energy can be written as

\begin{equation} \label{eq: isr_initial}
    E_{ISR}(f) = \int_{f(X_1)\cap X_2} (I_1\circ f^{-1} - I_2)^2 + \beta \frac{1}{2}\int_{X_1} |\nabla \mu(f)|^2,
\end{equation}
where $\beta > 0$ is a parameter that controls the strength of the geometric smoothing. We next describe the admissible set for the energy minimization problem.

\subsection{Choice of $\mathcal{A}$}
Other requirements on the registration map may be imposed. Mathematically, we need to design a suitable admissible set $\mathcal{A}$ where $f$ should lie in. In our case, we define our admissible set to be the intersection of three sets of mappings $\mathcal{A}=\mathcal{L}\cap \mathcal{S} \cap \mathcal{B}$. More specifically,
\begin{equation}
\begin{split}
	\mathcal{L} &=\{f:X_1 \to \mathbb{C}: f(p_i) = q_i, i=1,2,...,n\};\\
	\mathcal{S} &=\{f:X_1 \to \mathbb{C}: Df \in \mathcal{M}_{K_1,K_2}\};\\
	\mathcal{B}&=\{f:X_1 \to \mathbb{C}: ||\mu(f)||_{\infty}<1\}.
\end{split}	
\end{equation}
\noindent where $p_i\in X_1$ and $q_i\in X_2$ are the corresponding landmarks on $X_1$ and $X_2$ respectively, $\mathcal{M}_{K_{1},K_{2}}=\{M:\,K_{2}\leq\sigma_{2}(M)\leq\sigma_{1}(M)\leq K_{1}\}$, $\sigma_1(M)$ and $\sigma_2(M)$ are the largest and smallest singular values of $M$ respectively.
$\mathcal{L}$ encodes the landmark constraints. $\mathcal{S}$ aims to control the area distortion of $f$. Note that this is necessary, since the Beltrami coefficient of $f$ encodes the relative strength of stretching effect of $f$, but not the absolute scale distortion of the mapping. 
% Although the scale can be determined once any of the uniqueness-implying data is applied to solve the equation, such as the boundary value or the minimal condition, there are difficulties with such approaches.  In the case of boundary value data, it is often impractical to assume they are available. In the case of minimal condition, how it influences the scale under perturbation of the Beltrami coefficients can be difficult to handle, since one does not have the optimal Beltrami coefficient initially.
Finally, $\mathcal{B}$ aims to require the Beltrami coefficient of $f$ to have supremum norm strictly less than $1$, which enforces the local bijectivity of $f$ as discussed in Section \ref{sec:3}. 

\subsection{Applying splitting method to $E_{ISR}$}
Because $\mu(f)$ is a differential quantity of $f$, it is difficult to obtain the gradient direction of the $E_{reg}$ energy directly. In other words, one has to solve a partial differential equation to recover $f$ from $\mu(f)$. This will be explained below in Section \ref{sec:lbs}. We thus tackle each term separately using the following relaxation:
\begin{equation} \label{eq: isr_split}
    E_{ISR}(f,\nu) = \int_{f(X_1)\cap X_2} (I_1\circ f^{-1} - I_2)^2 + \frac{\alpha}{2}\int_{X_1} |\mu(f) - \nu|^2 + \frac{\beta}{2}\int_{X_1} |\nabla \nu|^2,
\end{equation}
where $\alpha,\beta>0$ are parameters, and the energy minimization is subject to $f\in \mathcal{A}$. The term $\frac{\alpha}{2}\int_{X_1} |\mu(f) - \nu|^2$ guarantees that our geometric smoothing preserves the geometric structures of $f$ as well as possible. Our splitting scheme thus passes in between the $f$-subproblem and $\nu$-subproblem, namely optimizing the coordinate function of the mapping $f$ and optimizing the Beltrami coefficients, which are connected via solving the Beltrami equation. We shall describe the splitting algorithm in Section \ref{sec:isr_algo}.

\section{Main algorithms} \label{sec:5}
In this section, we explain our algorithm to solve the optimization problem (\ref{eq: isr_split}) in details. In a high level overview, our algorithm involves two crucial steps. The first is the projection into the constraint set $\mathcal{S}$, which is about controlling the principal distortions of the mapping and is summarized in Algorithm \ref{algo:bdsv}. The second is the geometric smoothing step, which combines with the constraint yields the free boundary quasiconformal deformation algorithm, summarized in Algorithm \ref{algo: freeboundary}. Our final Algorithm \ref{algo: isr} that solves the ISR problem is obtained by combining intensity matching with the free boundary quasiconformal deformation algorithm.

In the following we first describe the projection onto $\mathcal{S}$. The projection onto $\mathcal{L} \cap \mathcal{B}$ can be taken care of during the projection onto $\mathcal{S}$ and during the minimization process of the regularization term $E_{reg}(f)$, which is explained in Section \ref{sec:freeboundary}. Our proposed algorithm for solving the ISR problem is derived by introducing the intensity matching subproblem into our deformation model, explained in Section \ref{sec:isr_algo}. Solving Beltrami equations is briefly introduced in Section \ref{sec:lbs}, whose details are referred to \cite{lam2014landmark,qiu2018parametrizing}.
 
\subsection{Projection of $f$ onto $\mathcal{S}$}\label{projectionstep}
In this subsection, we discuss the projection of $f$ onto $\mathcal{S}$ in details. We use an iterative scheme to project $f$ onto $\mathcal{S}$. Given a mapping $f$, its differential $Df$ may not lie in $\mathcal{M}_{K_1,K_2}$. We first look for a matrix $\mathcal{P}_{K_1,K_2}(Df)$ that solves the following minimization problem:
\begin{equation}\label{projmin}
	\mathcal{P}_{K_1,K_2}(Df)={\bf argmin}_{M\in \mathcal{M}_{K_1,K_2}} ||M- Df||_F^2
\end{equation}
where $\|\cdot\|_F$ denotes the Frobenius norm. 
Solving this minimization problem is equivalent to finding a projection of $Df(z)$ onto the space of matrices $\mathcal{M}_{K_1,K_2}$, whose singular values are bounded by $K_1$ and $K_2$. The solution of the problem can be given explicitly as follows.

\begin{theorem}\label{theorem1}
The unique solution of the minimization problem (\ref{projmin}) is given by
\[
\mathcal{P}_{K_1,K_2}(Df)=U\begin{pmatrix}\min(\sigma_{1}(Df),K_{1})\\
 & \max(\sigma_{2}(Df)),K_{2})
\end{pmatrix}V^{T},
\]
where $U,V$ are the rotation matrices such that $Df=U\begin{pmatrix}\sigma_{1}(Df)\\
 & \sigma_{2}(Df)
\end{pmatrix}V^{T}$. 
\end{theorem}
\begin{proof} 
This is related to the general two-sided Procrustes problem. Suppose $A_1$ and $A_2$ are two $n\times n$ matrices. $A_1 = U_1 \Sigma_1 V_1^T$ and $A_2 = U_2 \Sigma_2 V_2^T$ are the singular value decompositions of $A_1$ and $A_2$ respectively. Consider
\[
E_{P1} (Q_1,Q_2) = ||Q_1^TA_1 Q_2 - A_2||_F^2,
\]
\noindent where $Q_1$ and $Q_2$ are $n\times n$ orthogonal matrices. Then, the minimizers $Q_1^*,Q_2^*$ of $E_{P1}$ satisfy:
\begin{equation}\label{procruste}
U_1 = Q_1^* U_2 \Pi \text{ and } V_1 = Q_2^* V_2 \Pi,
\end{equation}
\noindent where $\Pi$ is the permutation matrix that minimizes $Tr(\Sigma_2\Pi\Sigma_1\Pi)$.

We now consider our original minimization problem (\ref{projmin}). Let $M = U_M \Sigma V_M^T$ be the singular value decomposition of $M$. Then, we observe that:
\[
|| U\begin{pmatrix}\sigma_{1}(Df)\\
 & \sigma_{2}(Df)
\end{pmatrix}V^{T} - U_M \Sigma V_M^T||_F^2 = || \begin{pmatrix}\sigma_{1}(Df)\\
 & \sigma_{2}(Df)
\end{pmatrix} - Q_1^T \Sigma Q_2||_F^2,
\]
\noindent where $Q_1=U_M^T U$ and $Q_2=V_M^T V$ are orthogonal matrices. Thus, the minimization problem 
(\ref{projmin}) is equivalent to minimizing with respect to $Q_1, \Sigma, Q_2$:
\[
E_{P2}(Q_1,\Sigma,Q_2) = || \begin{pmatrix}\sigma_{1}(Df)\\
 & \sigma_{2}(Df)
\end{pmatrix} - Q_1^T \Sigma Q_2||_F^2,
\]
where $Q_1, Q_2$ are orthogonal matrices and $\Sigma$ is a diagonal matrix with non-negative entries. Fixing a diagonal matrix $\Sigma$, we consider the minimization problem over $(Q_1, Q_2)$ of $E_{P2}(Q_1,\Sigma,Q_2)$. Then, the minimizer must satisfy $I=Q_1^*\Pi$ and $I=Q_2^*\Pi$ according to (\ref{procruste}). Thus, for any orthogonal matrices $Q_1$ and $Q_2$ together with any diagonal matrix $\Sigma$,
\begin{equation*}
\begin{split}
	E_{P2}(Q_1,\Sigma,Q_2) &= || \begin{pmatrix}\sigma_{1}(Df)\\
 & \sigma_{2}(Df)
\end{pmatrix} - Q_1^T \Sigma Q_2||_F^2\\
& \geq || \begin{pmatrix}\sigma_{1}(Df)\\
 & \sigma_{2}(Df)
\end{pmatrix} - \Pi \Sigma \Pi||_F^2\\
&\geq || \begin{pmatrix}\sigma_{1}(Df)\\
 & \sigma_{2}(Df)
\end{pmatrix} - 
\overline{\Sigma}||_F^2\\
&= E_{P2}(I,\overline{\Sigma},I),
\end{split}
\end{equation*}
\noindent where $\overline{\Sigma}$ is the minimizer of $|| \begin{pmatrix}\sigma_{1}(Df)\\
 & \sigma_{2}(Df)
\end{pmatrix} - D||_F^2$ over all diagonal matrix $D$ with diagonal entries bounded by $K_1$ and $K_2$. Obviously, $\overline{\Sigma}$ can be written explicitly as:
\[
\overline{\Sigma} = \begin{pmatrix}\min(\sigma_{1}(Df),K_{1})\\
 & \max(\sigma_{2}(Df)),K_{2})
\end{pmatrix}.
\]
\noindent We conclude that: $M^* =UI\overline{\Sigma}IV^T = U\overline{\Sigma}V^T$ is the minimizer of the problem (\ref{projmin}). This completes the proof.
\end{proof}

\bigskip

We call $\mathcal{P}_{K_1,K_2}$ the projection operator of a matrix onto the space $\mathcal{M}_{K_{1},K_{2}}$. In this way we can effectively control the scaling effect of the mapping at each point.  For simplicity we shall use $K_1,K_2$ uniformly on the domain, though other choices are clearly possible. We next look for a mapping $g:X_1\to \mathbb{C}$, whose differential is closely resemble to $\mathcal{P}_{K_1,K_2}(Df)$:
\begin{equation}
	E_{\mathcal{P}}(g) = \int_{X_1} ||Dg-\mathcal{P}_{K_1,K_2}(Df) ||_F^2.
\end{equation}

We can separate the above minimization problem for each coordinate function of $g = (u,v)$, and therefore we are left with two Poisson equations as their Euler-Lagrange equations:
\begin{equation} \label{eq:poisson}
    \begin{cases}
    -\Delta u &= -\nabla \cdot m_1 \\
    -\Delta v &= -\nabla \cdot m_2
    \end{cases},
\end{equation}
where $ \mathcal{P}_{K_1,K_2}(Df)(z) = (m_1(z), m_2(z))$ is a $2\times2$ matrix for each $z\in X_1$. The landmark constraints can be incorporated here via back substitution into \eqref{eq:poisson}. In this way, $f$ is constrained to lie in $\mathcal{L}$. We repeat this procedure iteratively. More precisely, given a map $f^{(i)}$, we compute $\mathcal{P}_{K_1,K_2}(Df^{(i)})$. We then compute $f^{(i+1)}$ that minimizes $E_\mathcal{P}$. The iterative algorithm can be summarized in Algorithm \ref{algo:bdsv}.

\begin{algorithm}
\caption{Projection onto $\mathcal{S}$}
\begin{algorithmic}
\STATE{{\bf Inputs: } Identity mapping $f^{(1)}$, iteration number $N$, bounds on singular values  $0<K_2 \leq K_1$, landmark correspondences $(p_i, q_i)$.}
\STATE{{\bf Output: } deformed mapping $f^{(N)}$.}

\hrulefill
\FOR{$i = 1, ..., N$} 
\STATE{Compute $Df^{(i)}$ and its singular value decomposition on each triangle.}
\STATE{Compute $\mathcal{P}_{K_1, K_2}(Df^{(i)})$ on each triangle.}
\STATE{Solve equations \eqref{eq:poisson} with landmark constraints to get $f^{(i+1)}$.}
\ENDFOR
\end{algorithmic}  \label{algo:bdsv}
\end{algorithm}

\subsection{Linear Beltrami Solver} \label{sec:lbs}
Solving Beltrami equation with given coefficient function $\mu$ is central to our optimization problem.
Solvers of the equation, with or without boundary data, are developed in the previous work \cite{lui2013texture,qiu2018parametrizing}, and therefore we briefly only introduce the main components here and refer the reader to the references for details.

The key is to note the relation between the Beltrami equations and {\it div}-type second order elliptic systems. 
In our case, we decouple the complex equation into its real and imaginary parts $u$ and $v$ respectively, and solve the following system of equations with Dirichlet boundary condition
\begin{equation}
\begin{cases}
-\nabla\cdot(A\nabla u(z))  =0 \, \text{ in } int(X_1)\\
-\nabla\cdot(A\nabla v(z)) =0 \, \text{ in } int(X_1)\\
u = u_{0} \, \text{ on } \partial X_1 \\
v = v_{0} \, \text{ on } \partial X_1
\end{cases}\label{eq:divAgrad}
\end{equation}
where $A= \frac{1}{1-|\mu|^2}\begin{bmatrix}(\rho-1)^{2}+\tau^{2} & -2\tau\\-2\tau & (1+\rho)^{2}+\tau^{2}\end{bmatrix}$, $\mu = \rho + i\tau$, and $int(X_1)$ is the interior of $X_1$.
In practice, the operator $\nabla\cdot(A\nabla)$ is discretized using linear finite element method on triangular meshes, which has the same form with the well-known cotangent formula for the discrete Laplacian \cite{qiu2018parametrizing}. We call such a solver the {\bf L}inear {\bf B}eltrami {\bf S}olver. We also denote the reconstructed quasiconformal map $f$ from $\mu$ by $f = {\bf LBS}(\mu,\{(p_i,q_i)\}_{i=1}^L)$, where $\{(p_i,q_i)\}_{i=1}^L$ denotes the prescribed corresponding landmark constraints.

\subsection{Free boundary quasiconformal deformation}\label{sec:freeboundary}
Before solving the registration model \eqref{eq: isr_split}, we describe a simpler algorithm for free boundary deformation with controlled distortions without considering intensity matching, namely
\begin{equation} \label{eq:freeboundary_deform_energy}
    E_{fbd}(f, \nu) = \frac{\alpha}{2}\int_{X_1} |\mu(f) - \nu|^2 + \frac{\beta}{2}\int_{X_1} |\nabla \nu|^2,
\end{equation}
subject to the constraint that $f\in \mathcal{A}$. The free boundary deformation does not require the source domain to be mapped to another target domain with fixed geometry. As such, surface registration problem, of which a natural bijection between two surfaces does not exist, can be handled using our formulation. 

The algorithm goes as follows. Suppose a quasiconformal map $f^{(k)}$ is obtained at the $k$-th iteration.
Firstly, we look for a quasiconformal map with a controlled scale distortion. It can be achieved by requiring the quasiconformal map to lie in  $\mathcal{S}$. Hence, the projection step as described in subsection (\ref{projectionstep}) is applied on $f^{(k)}$ to obtain a new map $g^{(N_1, 1)}$, where $N_1$ is the iteration number of the Algorithm \ref{algo:bdsv}. 
Secondly, fixing $f=g^{(N_1, 1)}$, we proceed to solve the $\nu$-subproblem, namely the geometric smoothing process. Denote the Beltrami coefficient of $g$ by $\mu_{g^{(N_1, 1)}}$. To enforce $g^{(N_1, 1)} \in \mathcal{B}$, we require that the $L^{\infty}$ norm of the Beltrami coefficient is strictly less than $1$. It ensures the bijectivity of the deformation map. We apply the following simple thresholding method:
\begin{equation} \label{eq:6}
\mu'_{g^{(N_1, 1)}}(z)=\begin{cases}
\mu_{g^{(N_1, 1)}}(z) & \text{ if }|\mu_{g^{(N_1, 1)}}(z)|<1\\
0 & \text{ otherwise}
\end{cases}.
\end{equation}
After thresholding, we proceed to update $\nu$ via the following iterative scheme:
\begin{equation} \label{descentnu}
\nu^{(j+1)} = (1- \alpha)\nu^{(j)} + \alpha \mu_{g^{(N_1, 1)}}' + \beta\Delta\nu^{(j)},
\end{equation}
where $\Delta$ is the Laplace-Beltrami operator of the domain and $\nu^{(0)}= \mu'_{g^{(N_1,1)}}$. This is essentially performing gradient descent on $E_{fbd}$ with $f$ fixed. Here we have abused the notation for parameter $\beta$ to include the step size as well. We can keep the process for a few (denoted $M_2$) iterations. The associated deformation map can be reconstructed by {\bf LBS} to obtain a new map $f^{(k+1)}$ with Dirichlet boundary condition given by $g$, subject to the prescribed landmark constraints that $f^{(k+1)}(p_i) = q_i$ for $i=1,2,...,n$. We keep this process for $M_1$ iterations to obtain $g^{(N_1,M_1)}$ and set $f^{(k+1)} = g^{(N_1,M_1)}$. We run the above iterations to obtain a sequence of quasiconformal maps $\{f^{(k)}\}_{k=1}^{N}$. The algorithm is summarized as in Algorithm \ref{algo: freeboundary}.

\begin{remark}
The term ``free boundary" refers to the property that our model requires no boundary constraints and need only a few landmarks to guide the deformation. This property can be seen from the discretization of the Poisson equation \eqref{eq:poisson}, where the Laplacian matrix involved is of rank $n-1$, $n$ is the number of vertices \cite{sorkine2005laplacian}. The geometric smoothing process is applied in the interior of the mapping only. Thus when using {\bf LBS} to recover the mapping from its Beltrami coefficients, we can assume that the boundary inherits from the projection to $\mathcal{S}$ step. Therefore, our model is applicable to domains of arbitrary topology, in particular multiply connected domains. This is experimentally illustrated in the Example 7 in Section \ref{sec:6}.
\end{remark}

\begin{algorithm}
\caption{Free boundary quasiconformal deformation}
\begin{algorithmic}
\STATE{{\bf Inputs: } Static and moving meshes with landmark correspondences $\{(p_i, q_i)\}_{i=1}^n$; overall iteration number $N$;  Projection step iteration number $N_1$, bounds on singular values  $0<K_2 \leq K_1$;  smoothing parameters $\alpha, \beta >0$ and smoothing step iteration numbers $M_1, M_2$.}
\STATE{{\bf Output: } Optimal free boundary quasiconformal deformation $f^{(N)}$.}

\hrulefill
\STATE{Initialize $f^{1}$ to be the identity mapping}
\FOR{$i = 1, ..., N$} 
\FOR{$l = 1, ..., N_1$}
\STATE{Compute $Df^{(i)}$ and its singular value decomposition on each triangle.}\STATE{Compute $M^{(i)}=\mathcal{P}_{K_1, K_2}(Df^{(i)})$ on each triangle.}
\STATE{Solve equations \eqref{eq:poisson} with landmark constraints $\{(p_i, q_i)\}_{i=1}^n$ to get $g^{(l, 1)}$.}
\ENDFOR
\FOR{$j = 1, ..., M_1$}
\STATE{Compute the Beltrami coefficients of $g^{(N_1, j)}$ and applying the thresholding \eqref{eq:6}.}
\FOR{$k = 1, ..., M_2$}
\STATE{Run the iteration according to (\ref{descentnu}).}
\ENDFOR
\STATE{Set $\nu = \nu^{(M_2+1)}$.}
\STATE{Set $g^{(N_1, j+1)} = {\bf LBS}(\nu,\{(b_i, b_i)\}_i)$, where $b_i$'s are the boundary points of $g^{(N_1, j)}$}
\ENDFOR
\STATE{Set $f^{(i+1)} = g^{(N_1, M_1)}$.}
\ENDFOR
\end{algorithmic}  \label{algo: freeboundary}
\end{algorithm}

\subsection{A splitting scheme for ISR model (\ref{eq: isr_split})} \label{sec:isr_algo}

In this subsection, we propose an iterative scheme to solve the optimization problem (\ref{eq: isr_split}). Our strategy is to minimize $E_{ISR}$ in an alternating direction fashion. In essence, the iterative algorithm is a modification of Algorithm \ref{algo: freeboundary} by adding a step to minimize the intensity mismatching to solve the ISR problem.

Suppose a mapping $f^{(k)}$ is obtained at the $k$-th iteration. 
Again, to enforce the deformation map to lie in $\mathcal{S}$, the projection step as described in subsection \ref{projectionstep} is carried out to obtain a new map $g^{(N_1,1)}$.

To minimize $E_{ISR}$, we first consider the descent direction of $E_{fid}$. Recall that our intensity registration problem aims to minimize the following loss function over $f$
\[
\widetilde{E}(f) = \int_{X_2} |I_1\circ f^{-1}-I_2|^{2}\cdot  \chi_{f(X_1)\cap X_2},
\]
where $\chi_{f(X_1)\cap X_2}$ denotes the indicator function of the set $f(X_1)\cap X_2$.

Given fixed $g^{(N_1,j)}(X_1)\cap X_2$, we can minimize the intensity mismatching term by using the gradient descent method. Write the displacement vector from the last deformation as $\mathbf{u}$, and $I_{1}^{g^{(N_1,j)}} = I_{1}\circ{(g^{(N_1,j)})^{-1}}$, the intensity mismatching term can be rewritten as
\[
E_{fid}(\mathbf{u}) = \int_{X_2} |I_1^{g^{(N_1,j)}} \circ (\text{Id} - \mathbf{u})  - I_2  |^{2} \chi_{g^{(N_1,j)}(X_1)\cap X_2}.
\]
We approximate it by Taylor expansion around $\mathbf{u}= \mathbf{0}$, obtaining
\[
E_{fid}(\mathbf{u}) \approx \int_{X_2} |I_1^{g^{(N_1,j)}}  - I_2 - \mathbf{u}^T \nabla I_1^{g^{(N_1,j)}} |^{2} \chi_{g^{(N_1,j)}(X_1)\cap X_2}.
\]
Minimizing this energy using gradient descent directly may lead to unstable and non-diffeomorphic solutions. Regularization is often introduced:
\[
\tilde{E}_{fid}(\mathbf{u}) =\int_{X_2} |I_1^{g^{(N_1,j)}}  - I_2 - \mathbf{u}^T \nabla I_1^{g^{(N_1,j)}} |^{2} \chi_{g^{(N_1,j)}(X_1)\cap X_2} + \sigma_I|u|^2,
\]
where $\sigma_I$ is a related to the intensity difference of the two images. If we set $\sigma_I = \tau^2 (I_1^{g^{(N_1,j)}}  - I_2)^2$, the optimal displacement field can be obtained by the first order condition
\begin{equation} \label{eq: demons}
\mathbf{u}^* = \frac{(I_1^{g^{(N_1,j)}}  - I_2)\chi_{g^{(N_1,j)}(X_1)\cap X_2}}{|\nabla I_1^{g^{(N_1,j)}}|^2 + \tau^2 (I_1^{g^{(N_1,j)}}  - I_2)^2} \nabla I_1^{g^{(N_1,j)}}
\end{equation}
which is known as the Demons algorithm \cite{thirion1998image}. Recall that our algorithm aim to obtain a diffeomorphic registration. The above displacement field is often non-smooth and will result in non-diffeomorphic deformation mappings. To alleviate this issue, Gaussian filtering on the deformation field is often introduced \cite{pennec1999understanding,thirion1998image}. Note that other versions of the Demons algorithm, such as the one proposed in \cite{wang2005validation}, can also be used.
As a result, we can obtain a vector field ${\bf V} = (V_1, V_2)$ that gives a new map $\tilde{g} = g^{(N_1,j)} + {\bf V}$. Its associated Beltrami coefficient $\mu_{\tilde{g}}$ can be thus computed.

\begin{algorithm}[t]
\caption{Registration on inconsistent domains}
\begin{algorithmic}
\STATE{{\bf Inputs and parameters: } conformal parametrizations of the moving and static triangular meshes with intensity defined on vertices; overall algorithm iteration number $N$;  free boundary subproblem iteration number $N_1$, bounds on singular values  $0<K_2 \leq K_1$, landmark correspondences $(p_i, p_i')$; Intensity matching subproblem overall iteration number $M_1$; smoothing parameters $\alpha, \beta >0$ and smoothing iteration number $M_2$.}
\STATE{{\bf Output: } optimal deformation $f^{(N)}$ given landmarks and intensities in the sense of \eqref{eq: isr_split}}.

\hrulefill
\STATE{Initialize $f^{(1)}$ to be the identity mapping.}
\FOR{$i = 1, ..., N$} 
\FOR{$l = 1, ..., N_1$}
\STATE{Compute $Df^{(i)}$ and its singular value decomposition on each triangle.}
\STATE{Compute $\mathcal{P}_{K_1, K_2}(Df^{(i)})$ on each triangle.}
\STATE{Solve equations \eqref{eq:poisson} with landmark constraints to get $g^{(l,1)}$.}
\ENDFOR
\FOR{$j = 1, ..., M_1$}
\STATE{Obtain a vector field $V^{(j)}$ on $g^{N_1, j}(X_1)\cap X_2$ by matching intensities $\chi_{g^{(N_1, j)}(X_1)\cdot \cap X_2} I_{X}\circ (g^{(N_1, j)})^{-1}$ and $\chi_{g^{(N_1, j)}(X_1)\cap X_2} \cdot  I_{2}$ according to Equation \eqref{eq: demons}.}
\STATE{Apply $V^{(j)}$ to $g^{(N_1, j)}$ to get new deformation map $\tilde{g}$.}
\STATE{Compute the Beltrami coefficients of $\tilde{g}$ and applying the thresholding \eqref{eq:6}.}
%\FOR{$m = 1, ..., L$}
\STATE{Run $M_2$ iterations according to (\ref{descentnu}) to iteratively minimize $E_{reg}$ to get an update Beltrami coefficient $\nu$.}
%\ENDFOR
\STATE{Solve the Beltrami equation using the updated coefficient to get $g^{(M_1, j+1)}$.}
\ENDFOR
\STATE{Set $f^{(i+1)} = g^{(N_1, M_1)}$.}
\ENDFOR
\end{algorithmic}  \label{algo: isr}
\end{algorithm}
To enforce the deformation map lies in $\mathcal{B}$, we require that the supremum norm of the Beltrami coefficient is strictly less than 1. It ensures the bijectivity of the deformation map. We again apply the following simple thresholding method:
\begin{equation} \label{eq:6'}
\mu'_{\tilde{g}}(z)=\begin{cases}
\mu_{\tilde{g}}(z) & \text{ if }|\mu_{\tilde{g}}(z)|<1\\
0 & \text{ otherwise}
\end{cases}.
\end{equation}

After that, we proceed to minimize the regularization term. Geometric smoothing can be achieved using the equation (\ref{descentnu}) as described in Section \ref{sec:freeboundary}. This gives an updated Beltrami coefficient $\nu$.
The associated quasiconformal map can be reconstructed by {\bf LBS} to obtain $g^{(k,2)} = {\bf LBS}(\nu)$, subject to the prescribed landmark constraints. We repeat the process to iteratively minimize $E_{ISR}$ for $M$ iterations and set $f^{(k+1)} = g^{(N_1, M_1)}$. The main algorithm can now be summarized as in Algorithm \ref{algo: isr}.

\begin{remark}
Since in our formulation the energy is highly nonlinear in $f$, it would be worthwhile to explain why our algorithm is effective to minimize the energy. Note that in Algorithm \ref{algo: isr} we start with the free boundary deformation model where landmarks are matched. The result mapping will then be a good initialization for the intensity matching subproblem. This can be visualized in the energy plots of the experiments in Section \ref{sec:6}, where energy decreases very fast in the first few iterations, implying the splitting method works well in practice. 
\end{remark}

\section{Experimental results} \label{sec:6}
Before we proceed to the experimental results, several comments about the efficiency of our algorithm and hyper-parameter selection are in order. First of all, the projection operator $\mathcal{P}_{K_1, K_2}$ is independent for each triangle in the mesh, and thus can be parallelized. This is also true for intensity matching which can be parallelized for each vertex or pixel. The main computations are solving the linear equations \eqref{eq:poisson} and \eqref{eq:divAgrad}, both are very fast since the linear system involved is symmetric positive definite. For a mesh with $~5k$ vertices and $~9k$ triangles, in our un-optimized implementation using {\tt Matlab}, each inner iteration costs $~0.2s$ and $~0.5s$ on an 8-core Intel i7 machine, for the free boundary quasiconformal deformation model and the intensity matching step. Secondly, the hyper-parameters $\alpha, \beta$ generally give similarly results when in range $[0,0.1]$, depending how smooth the user wants the mapping to be. And we find the choice of iteration number $N_1, M_1, M_2$ from $1-5$ to achieve a good balance between computational cost and results. For most experiments demonstrated in this paper, the deformation has converged before $N=50$, mainly depending on the complexity of intensity matching. The choice for bounds $K_1 = 2, K_2 = 0.5$ in general works well except when the landmark induced deformation is too large, or can be adjusted if the user wants to put apriori constraint on the registration mapping. We also find our algorithm is fairly robust to choices of hyper-parameters in the Demons algorithm \cite{thirion1998image}. In the following we demonstrate results under various hyper-parameters, while different settings as described above should work similarly well. Our implementation is made public in GitHub\footnote{\url{https://github.com/sylqiu/incon_reg}}.

\subsection{Synthetic numerical experiments}
In this subsection, we test the performance of our proposed algorithms on two synthetic examples in the 2D domain. The first experiment is an ablation study that aims to demonstrate the effect of geometric smoothing. The second experiment aims to demonstrate the intensity matching in a pair of inconsistent image domains.

\bigskip

\noindent {\bf Example 1} {\it (The Effect of geometric smoothing)}: In this example, we test our algorithm for free boundary quasiconformal deformation (Algorithm 2), given landmark constraints as shown in the first row of Figure \ref{fig:gsmooth}. 
% We have chosen $\alpha=\beta=0.1$ and smoothing steps $M_1=30, M_2=10$. Bounds on singular values are chosen to be $K_1=5, K_2=0.2$. 
The result without geometric smoothing (that is, setting $M_1 = M_2 = 0$) is shown on the left in the second row, where one can clearly observe flipping and non-smooth singularities around the landmark points. The result with geometric smoothing is shown on the right in the second row, where the local deformations are well propagated around the landmark points and the obtained deformation is smooth. This shows the necessity and effectiveness of the geometric smoothing scheme.

\bigskip

\noindent {\bf Example 2} {\it (Registration of a pair of images)}: In this example, we test our intensity registration model (Algorithm 3) on an inconsistent pair of images of the letter ``A". 
In Figure \ref{fig:A_data}, the left on the first row shows the input images with their intensities to be matched, where the one on the left is the moving domain. Our goal is to register the left taller ``A" to the middle wider ``A". The given landmark correspondences are shown on the left in the second row. In this example, we set $\alpha=0.01, \beta = 0.01$, smoothing steps $M_1=1, M_2=10$; bounds $K_1 = 1.4, K_2 = 0.2$; free boundary subproblem for $N_1 = 1$ iterations; intensity subproblem for $20$ iterations and overall iteration $N=20$. The registered image from the moving image is shown on the right in the second row. It closely resembles to the target static image. The absolute difference in intensity after registration is shown on the right in the first row. This experimental result demonstrates that our algorithm is effective in finding an accurate registration map matching the intensities as well as finding the corresponding regions on the images. 

We also display the energy plot against the iteration number averaged per face of the mesh on the left of Figure \ref{fig:A_stat}, and the landmark error plot on the left. We can see that the algorithm successfully reduces the intensity and landmark mismatching errors.
 \begin{figure}
     \centering
     \subfloat{\includegraphics[width=6.5cm]{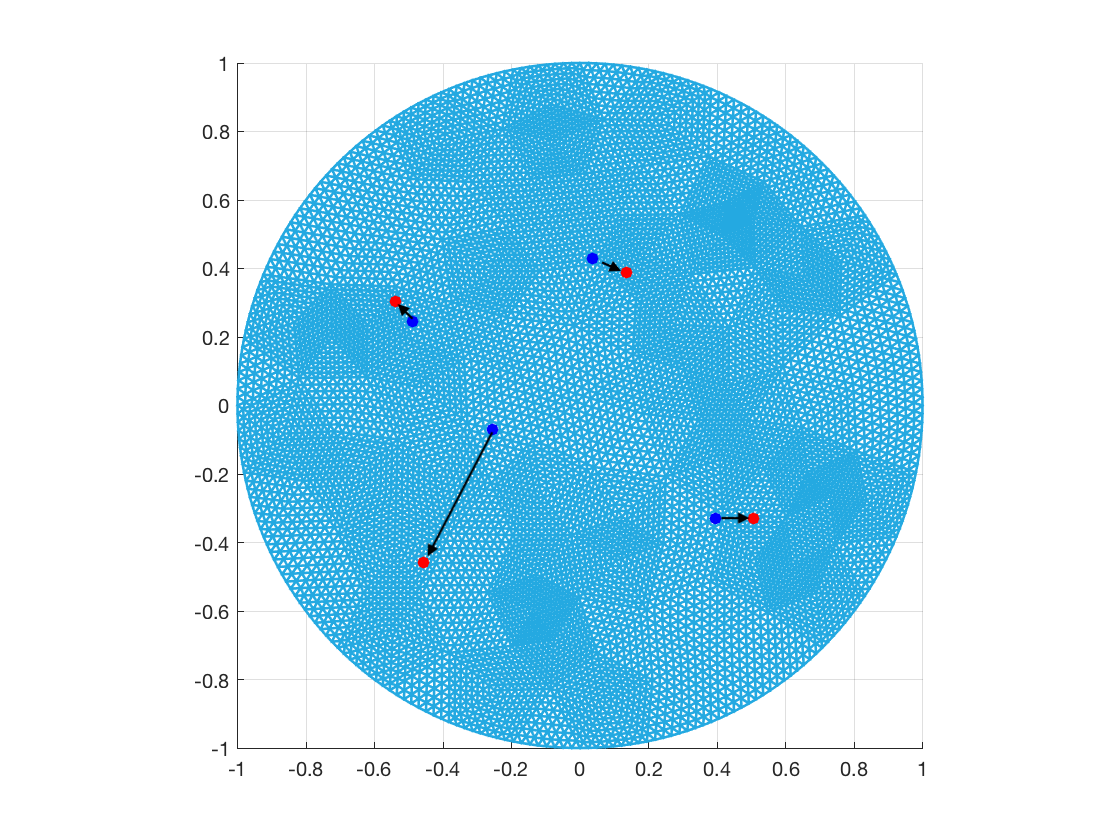}}\\
     \subfloat{\includegraphics[width=6.5cm]{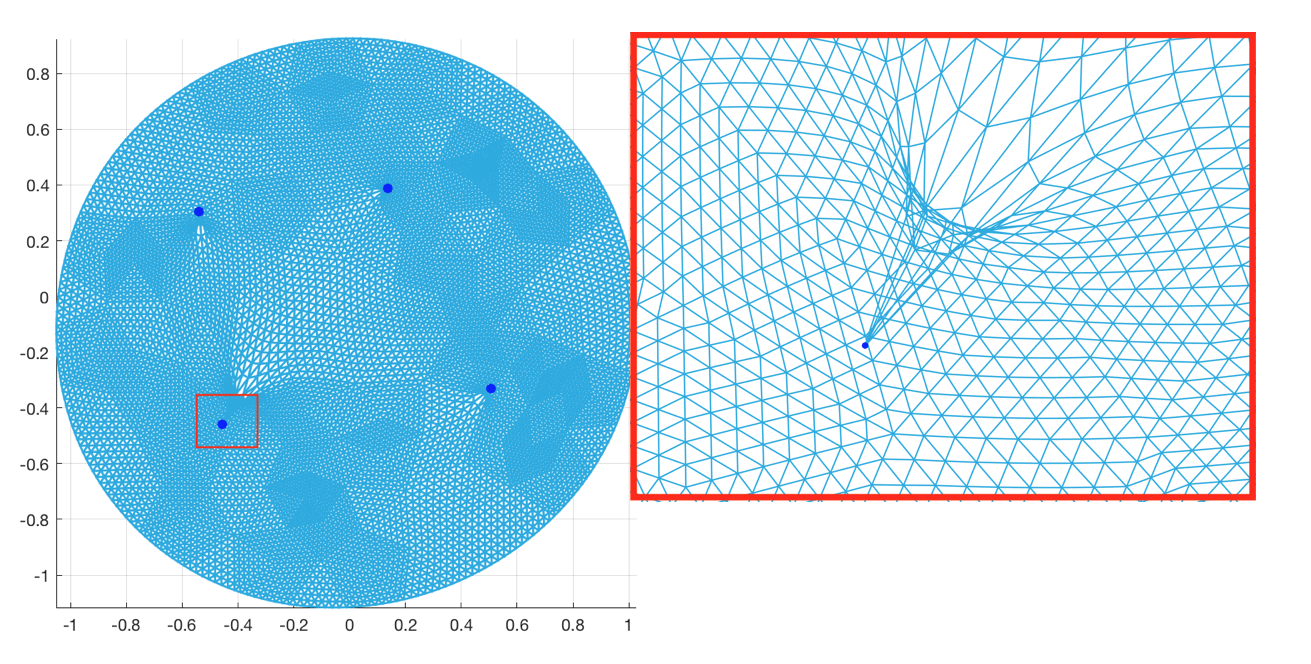}} 
     \subfloat{\includegraphics[width=6.5cm]{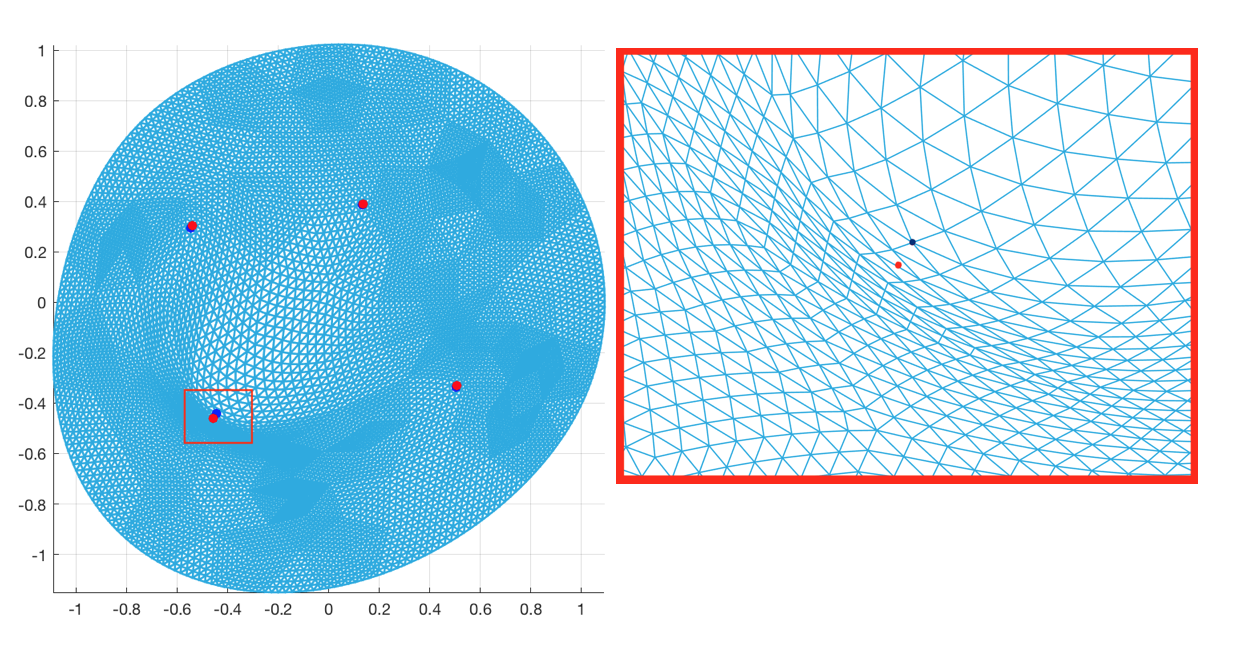}} 
     \caption{Effect of the geometric smoothing. The first row shows the input mesh. The landmark correspondences are displayed. We deform the input mesh so that the blue landmarks are matched to the red landmarks, using our proposed free boundary quasiconformal deformation algorithm. In the second row, the result on the left is obtained without geometric smoothing. The result on the right shows the output of the algorithm with geometric smoothing.}
     \label{fig:gsmooth}
 \end{figure}
 
  \begin{figure}
     \centering
     \subfloat{\includegraphics[height=4.5cm]{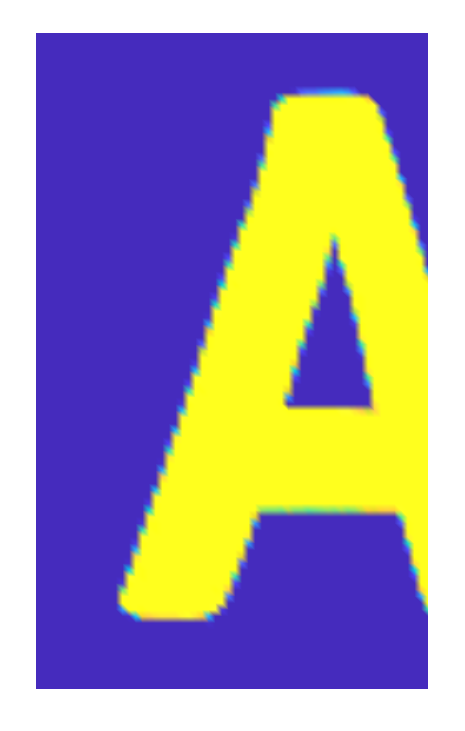}} \hspace{10pt}
     \subfloat{\includegraphics[height=4.5cm]{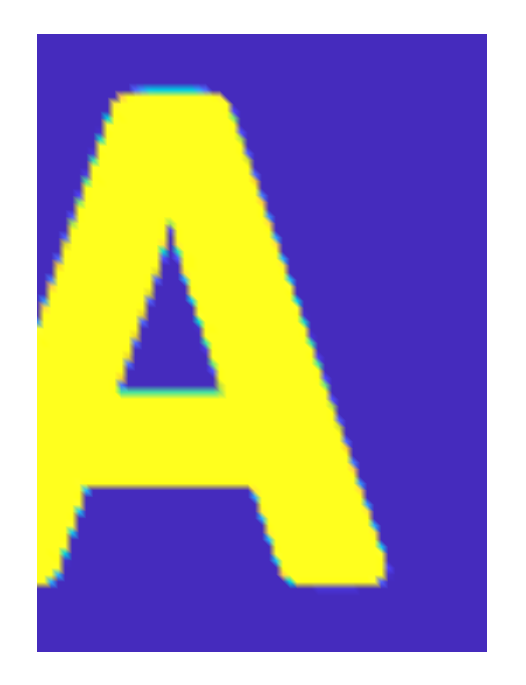}}\hspace{10pt}
     \subfloat{\includegraphics[height=4.5cm]{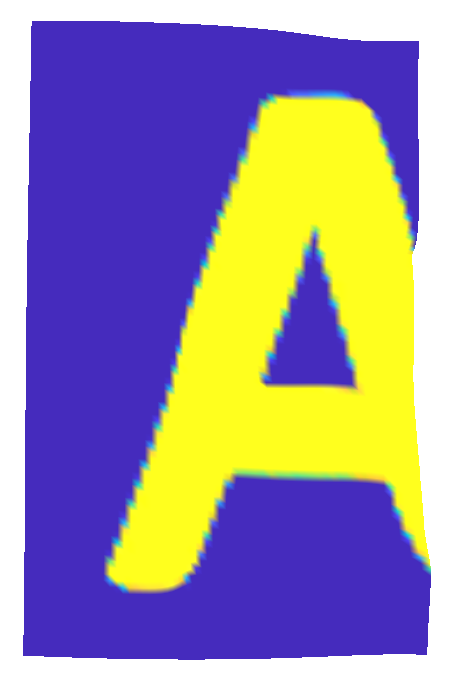}}
     \subfloat{\includegraphics[height=4.85cm]{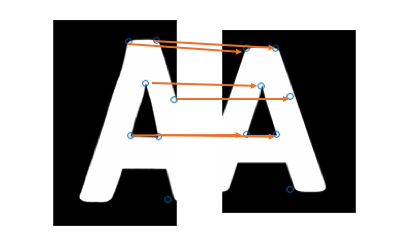}} \hspace{1pt}
    \subfloat{\includegraphics[height=4.5cm]{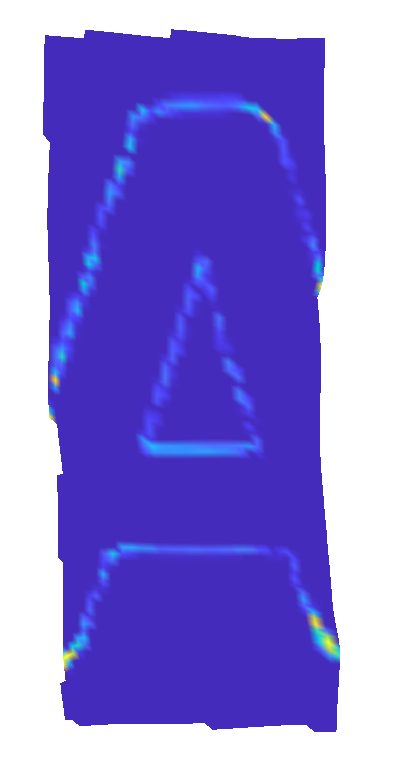}} 
     \subfloat{\includegraphics[width=18pt]{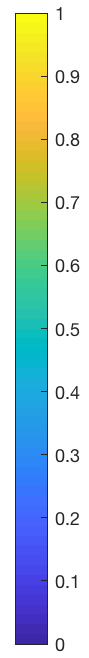}}\\

     \caption{Result of the inconsistent image registration experiment in Example 2. The left in the first row shows the input moving image. The middle in the first row shows the input static image. Our goal is to register the left taller ``A" to the middle wider ``A". The left in the second row shows the prescribed landmark correspondences. The registered image from the moving image is shown on the right in the first row. The right in the second row shows the absolute difference in intensity after registration.}
     \label{fig:A_data}
 \end{figure}

 \begin{figure}
     \centering
     \subfloat{\includegraphics[width=7.2cm]{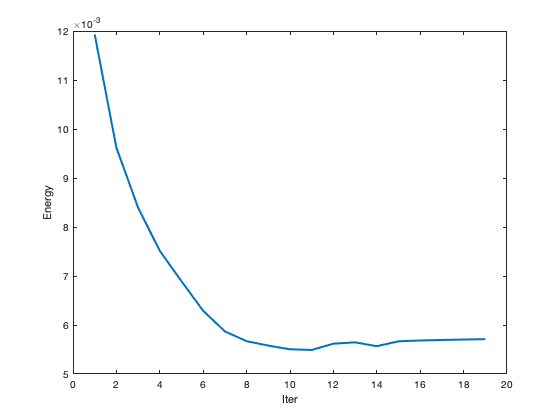}} 
     \subfloat{\includegraphics[width=7.2cm]{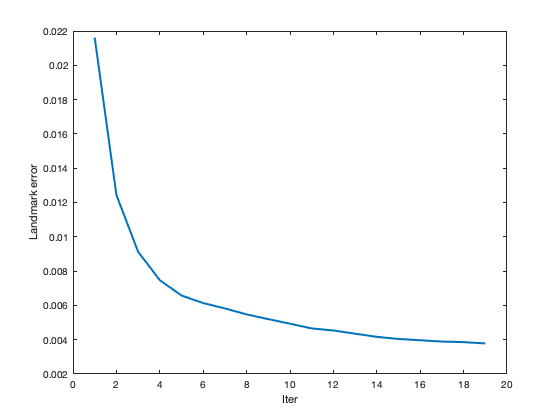}}
   
     \caption{Energy and landmark error plots for the image registration experiment against iteration number. Energy is averaged per triangle. The left shows the overall energy versus iterations. The right shows the landmark mismatching error versus iterations.}
     \label{fig:A_stat}
 \end{figure}

  \begin{figure}
     \centering
     \subfloat{\includegraphics[height=4.5cm]{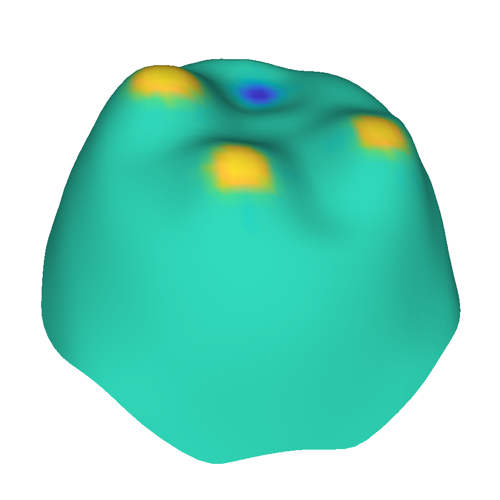}}
     \subfloat{\includegraphics[height=4.5cm]{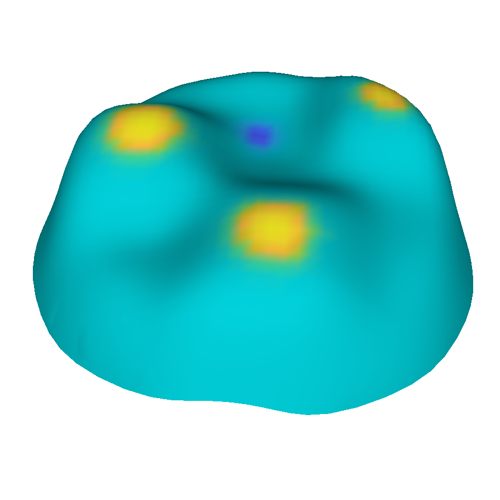}} \hspace{10pt}
     \subfloat{\includegraphics[height=4.5cm]{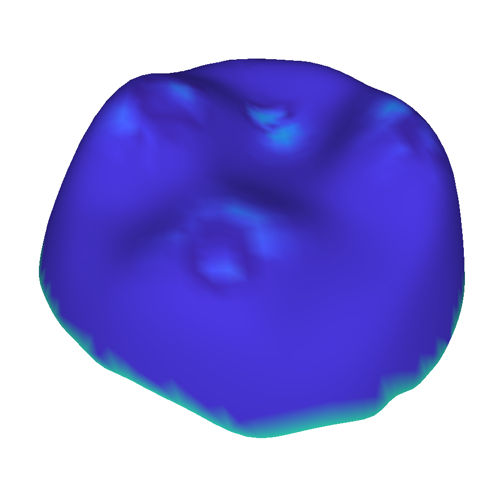}}
     \subfloat{\includegraphics[width=18pt]{figures/colorbar.png}}
     \caption{Surface registration for a pair of inconsistent tooth surfaces in Example 3. The left shows the input moving tooth surface. The middle shows the target static tooth surface. The right shows the difference of intensities on the registered surface.}
     \label{fig:teeth_data}
 \end{figure}
 
 \begin{figure}
     \centering
     \subfloat{\includegraphics[height=5.5cm]{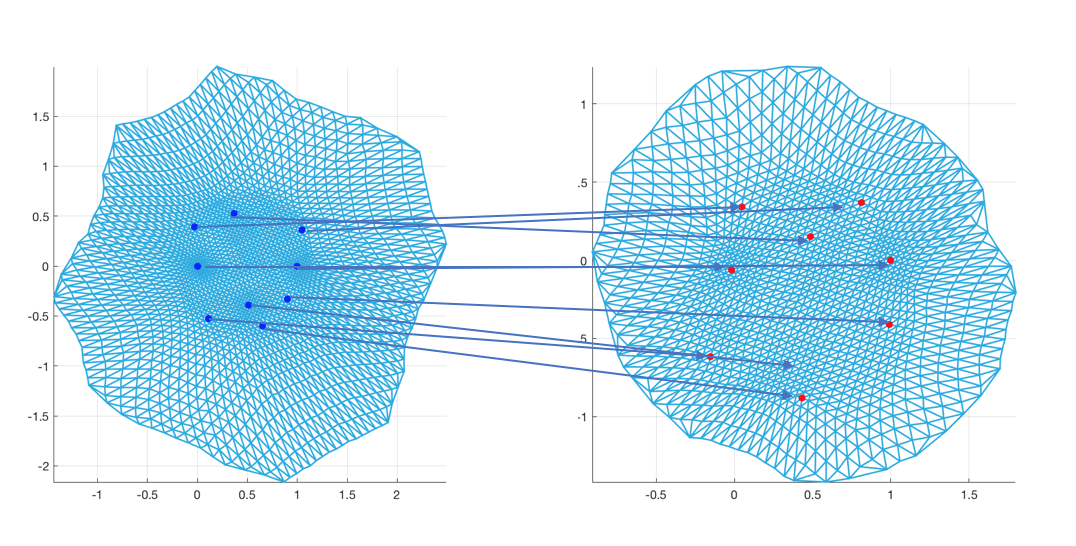}} \\
     \subfloat{\includegraphics[height=5.5cm]{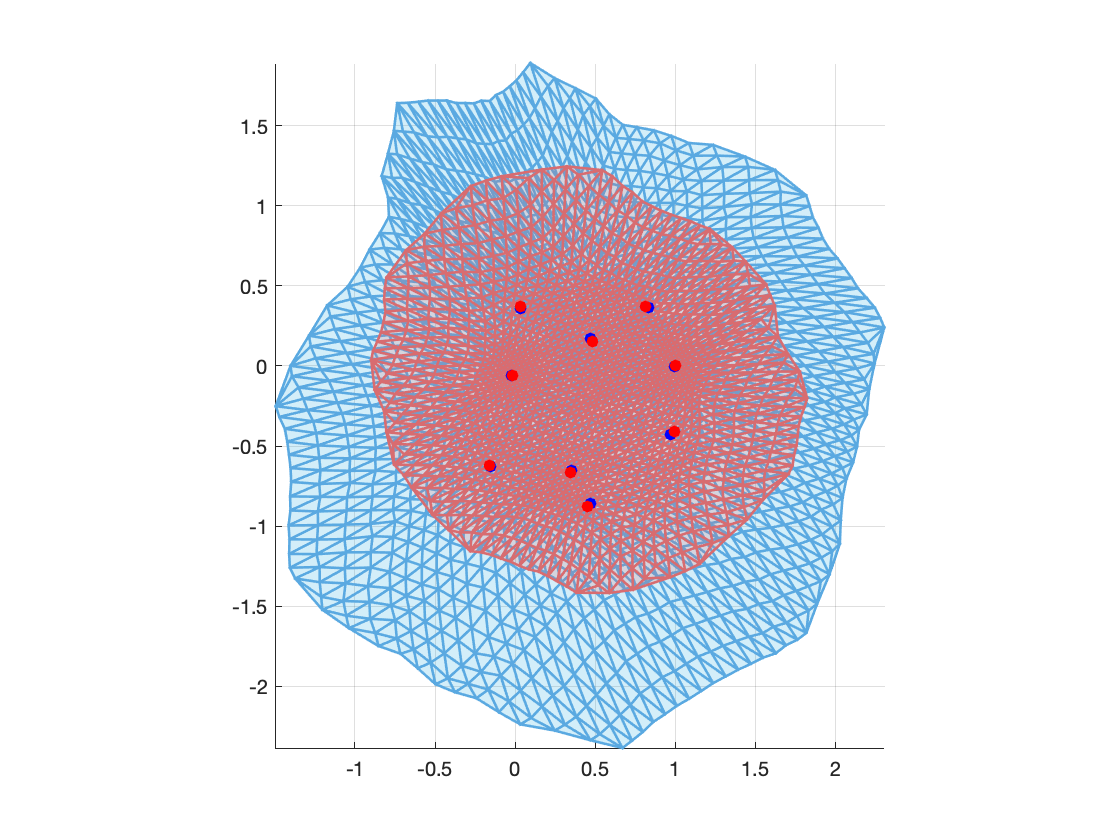}}
     \caption{The left in the first row shows the conformal parametrization of the moving tooth surface in Example 3. The right in the second row shows the conformal parametrization of the static tooth surface. The landmark correspondences in the 2D parameter domains are also displayed. The registration result in the 2D domain is shown in the second row. The blue mesh is transformed mesh from the moving mesh under the deformation map. The red mesh is the 2D mesh of the target surface under the conformal parametrization.}
     \label{fig:teeth}
 \end{figure}
 
 \begin{figure}
     \centering
     \subfloat{\includegraphics[height=4.5cm]{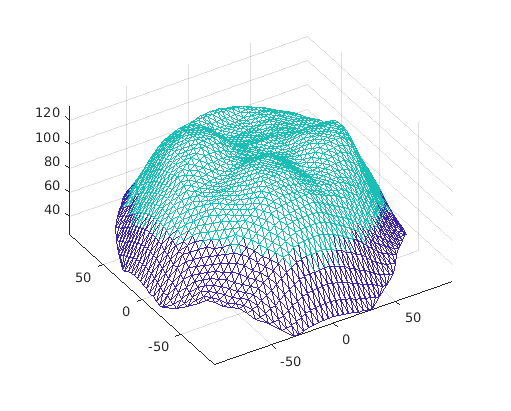}} \hspace{10pt}
     \subfloat{\includegraphics[height=4.5cm]{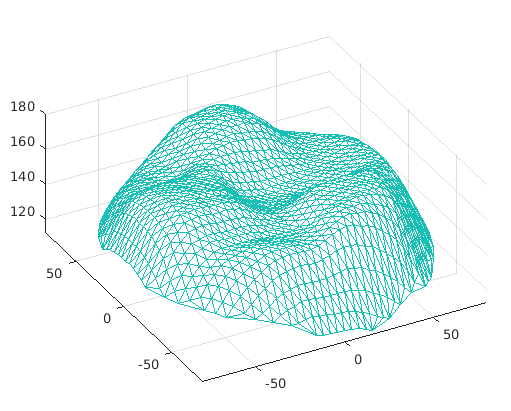}}
     \caption{Corresponding regions on the moving and static tooth surfaces in Example 3. The green region on left shows the corresponding region on the moving tooth surface. The green region on the right shows the corresponding region on the target tooth surface.}
     \label{fig:teeth_corr}
 \end{figure}
 
   \begin{figure}
     \centering
     \subfloat{\includegraphics[width=7.2cm]{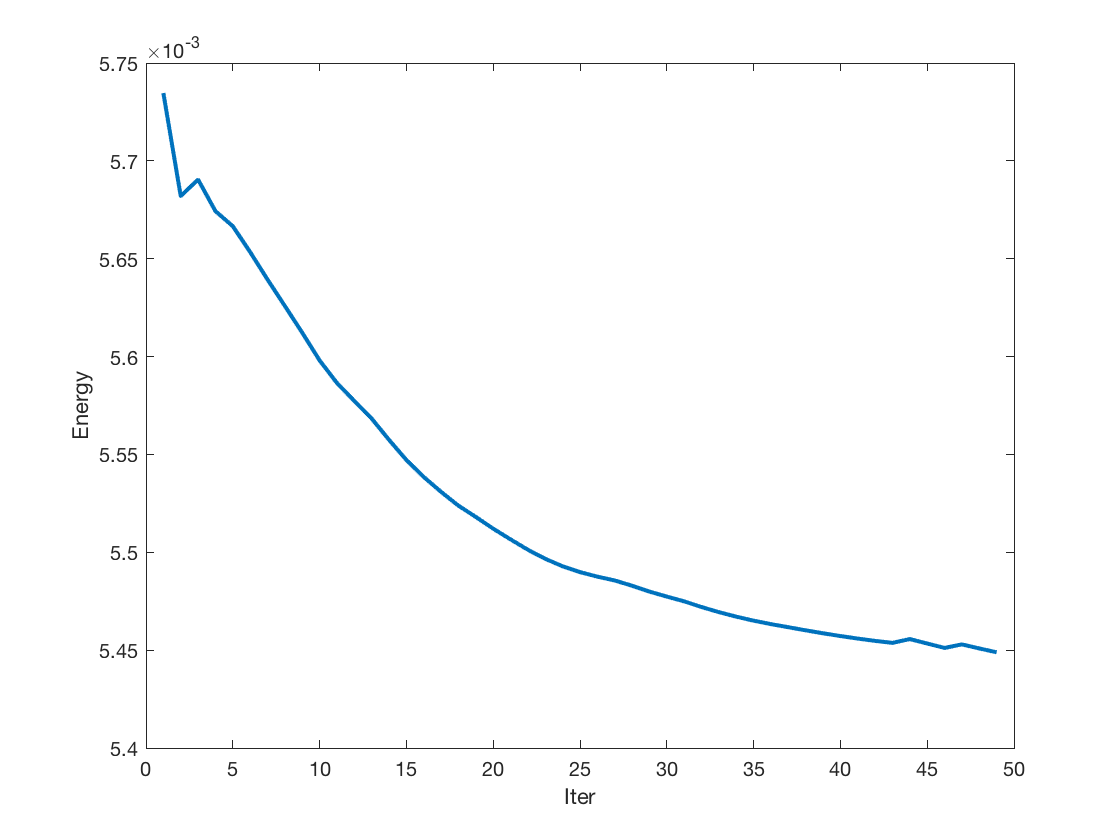}}\hspace{10pt}
     \subfloat{\includegraphics[width=7.2cm]{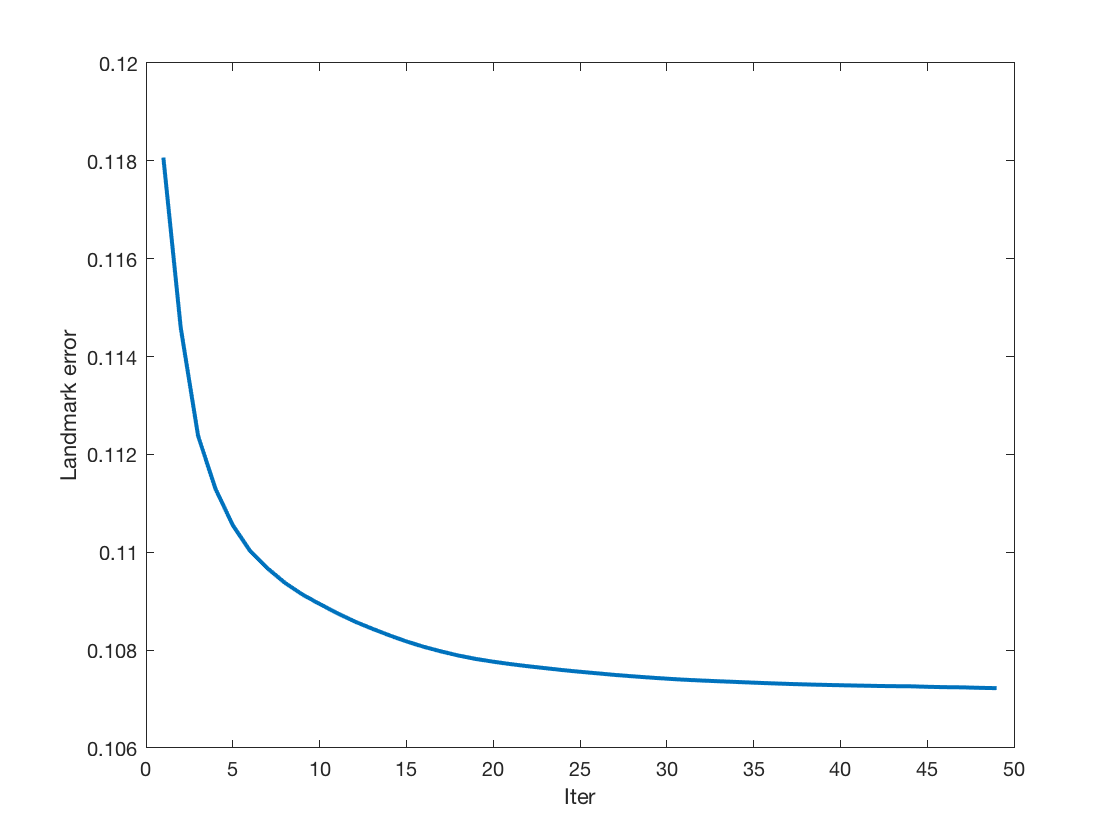}}
     \caption{Energy and landmark error plots for the teeth experiment in Example 3 against iteration number. Energy is averaged per face. The left shows the overall energy versus iterations. The right shows the landmark mismatching error versus iterations.}
     \label{fig:teeth_stat}
 \end{figure}

\subsection{Experiments on surfaces}
In this subsection, we test our registration model on surfaces.

\bigskip

\noindent {\bf Example 3} {\it (Registration of a pair of synthetic teeth surfaces)}: In this example, we test the registration model on an inconsistent pair of tooth surfaces. The two surfaces are not bijectively corresponding to each other. Only a partial subset of the source surface is in correspondence with a subset of the target surface.

In Figure \ref{fig:teeth_data}, the first and second columns shows the input surfaces with their curvatures to be matched, where the one on the left is the moving surface. 
We first perform conformal flattening \cite{levy2002least} of the two 3D meshes into 2D, as shown in the first row in Figure \ref{fig:teeth}. The landmark correspondences in the 2D parameter domains are also displayed. In this example, we have used $\alpha=0.01, \beta = 0.1$, smoothing steps $M_1=1, M_2=10$; bounds $K_1 = 1.2, K_2 = 0.8$; free boundary subproblem for $N_1 = 5$ iterations; intensity subproblem for $10$ iterations and overall iteration $N=50$. The registration result is shown in the last column of Figure \ref{fig:teeth_data}. It is the registered surface from the moving surface to the target static surface. The colormap on the surface is given by the curvature mismatching error. Note that the mismatching error is small, indicating the curvatures are accurately matched. The registration result in the 2D domain is shown in the second row of Figure \ref{fig:teeth}. The blue mesh is transformed mesh from the moving mesh under the deformation map. The red mesh is the 2D mesh of the target surface under the conformal parametrization. The intersection region of the two meshes is the region of correspondence amongst the two tooth surfaces.

We also display the energy plot against iteration number averaged per face of the mesh on the left of Figure \ref{fig:teeth_stat}, and the landmark error plot on the left. We can see that the algorithm successfully reduces the intensity and landmark matching errors.

 \begin{figure}
     \centering
     \subfloat{\includegraphics[height=4cm]{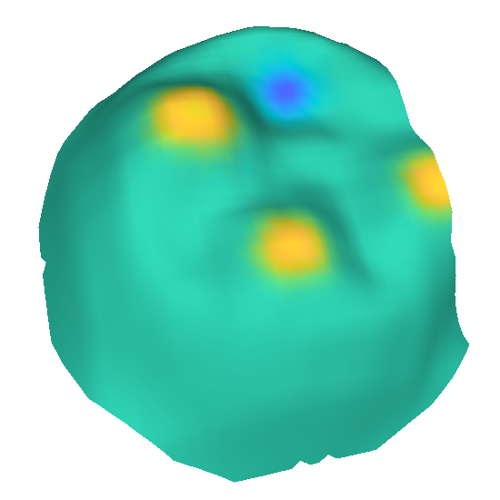}}\hspace{10pt}
     \subfloat{\includegraphics[height=4cm]{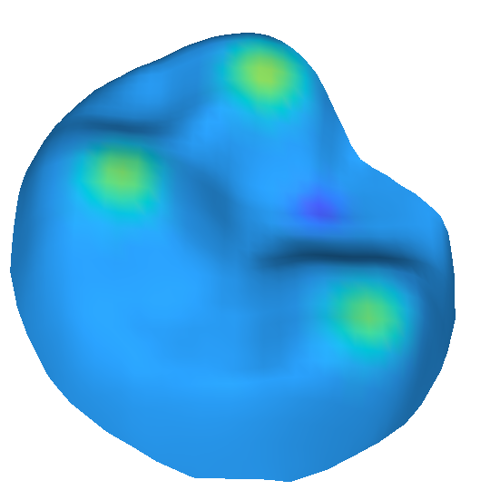}}
     \subfloat{\includegraphics[height=4cm]{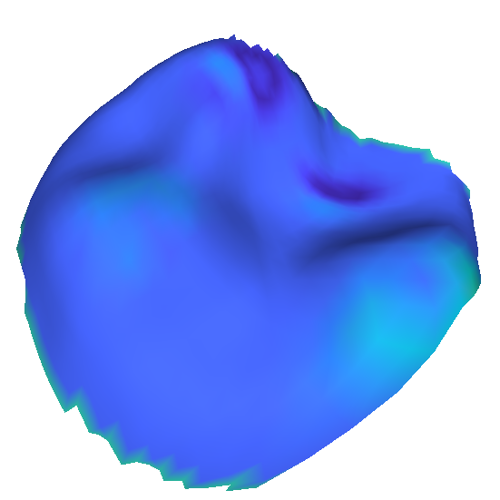}} 
     \subfloat{\includegraphics[width=18pt]{figures/colorbar.png}}
     \caption{Surface registration for another pair of inconsistent tooth surfaces in Example 4. The left shows the input moving tooth surface. The middle shows the target static tooth surface. The right shows the difference of intensities on the registered surface.}
     \label{fig:teeth2_data}
 \end{figure}

 \begin{figure}
     \centering
      \subfloat{\includegraphics[height=5.5cm]{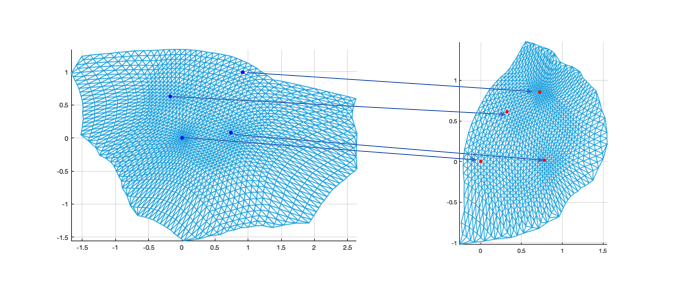}} \\
     \subfloat{\includegraphics[height=5.5cm]{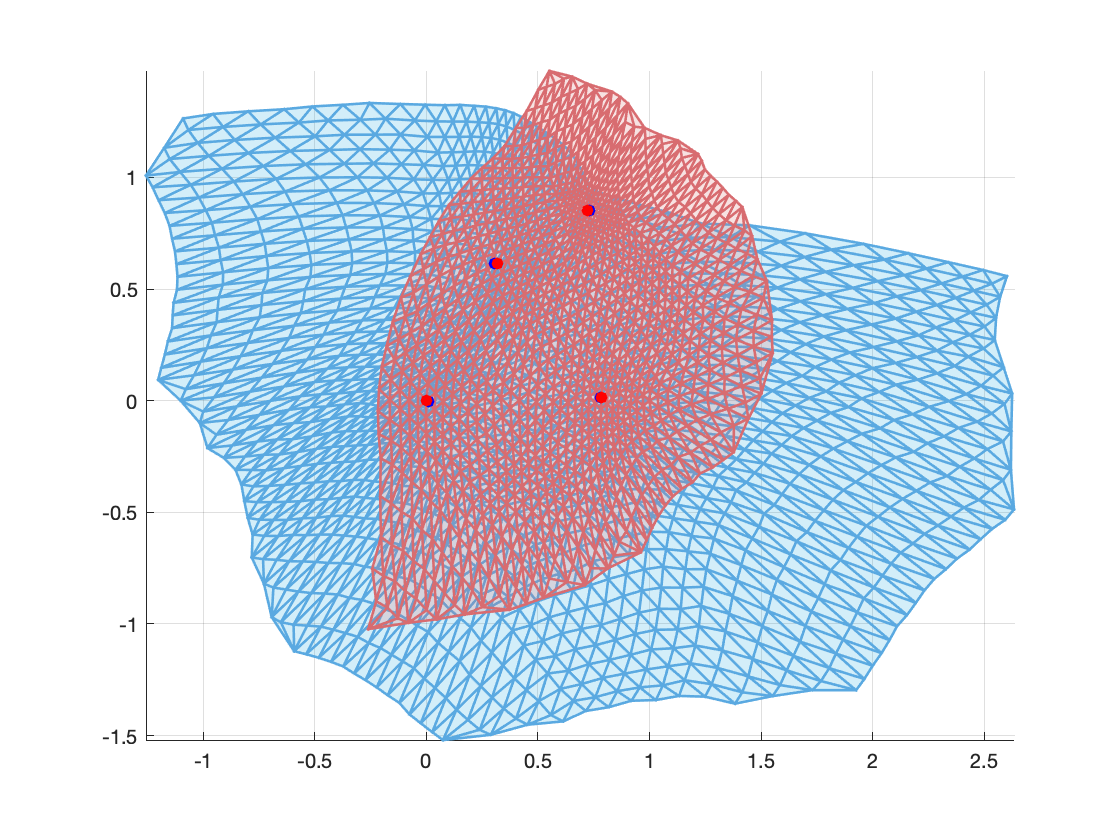}}
     \caption{The first row shows the conformal parametrizations of the moving and static tooth surfaces in Example 4 respectively. The landmark correspondences in the 2D parameter domains are also displayed. The registration result in the 2D domain is shown in the second row. The blue mesh is transformed mesh from the moving mesh under the deformation map. The red mesh is the 2D mesh of the target surface under the conformal parametrization.The intersection represents the corresponding region.}
     \label{fig:teeth2}
 \end{figure}
 
 \begin{figure}
     \centering
     \subfloat{\includegraphics[height=4.5cm]{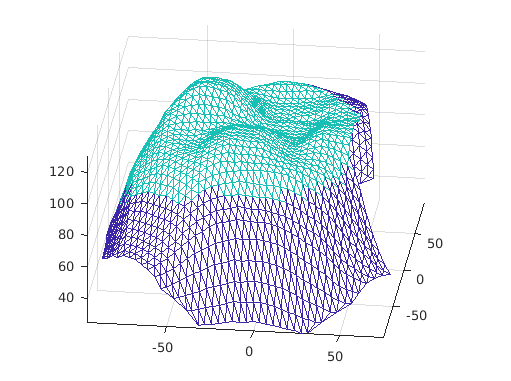}} \hspace{10pt}
     \subfloat{\includegraphics[height=4.5cm]{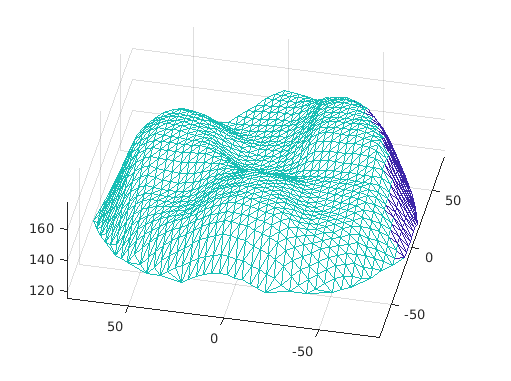}}
     \caption{Corresponding regions on the moving and static tooth surfaces in Example 4. The green region on left shows the corresponding region on the moving tooth surface. The green region on the right shows the corresponding region on the target tooth surface.}
     \label{fig:teeth2_corr}
 \end{figure}

\begin{figure}
     \centering
     \subfloat{\includegraphics[width=7.2cm]{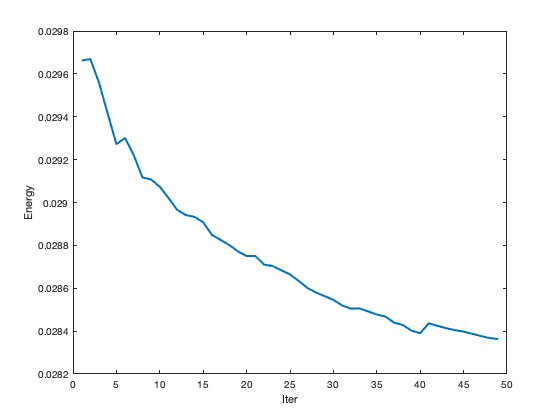}} \hspace{10pt}
     \subfloat{\includegraphics[width=7.2cm]{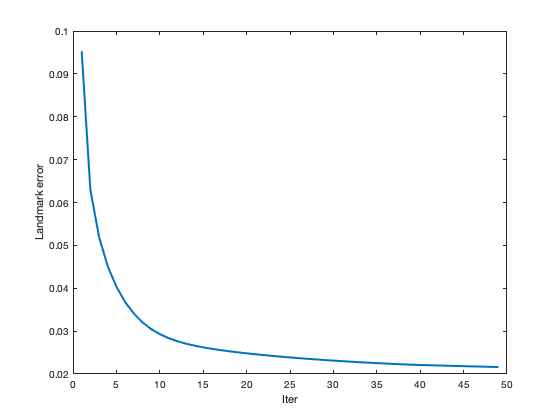}}
     \caption{Energy and landmark error plots for the tooth surface experiment in Example 4 against iteration number. Energy is averaged per face. The left shows the overall energy versus iterations. The right shows the landmark mismatching error versus iterations.}
     \label{fig:teeth2_stat}
 \end{figure}

  \begin{figure}
     \centering
     \subfloat{\includegraphics[height=4.5cm]{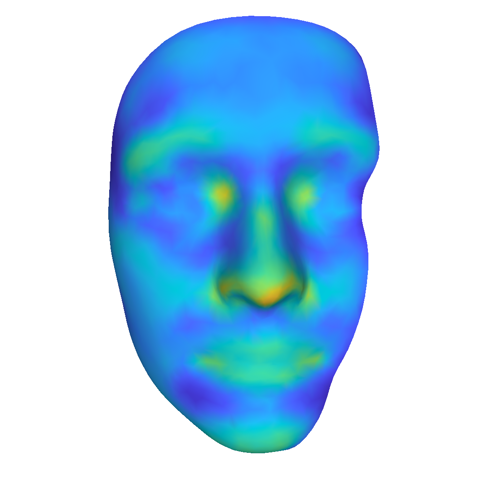}}
     \subfloat{\includegraphics[height=4.5cm]{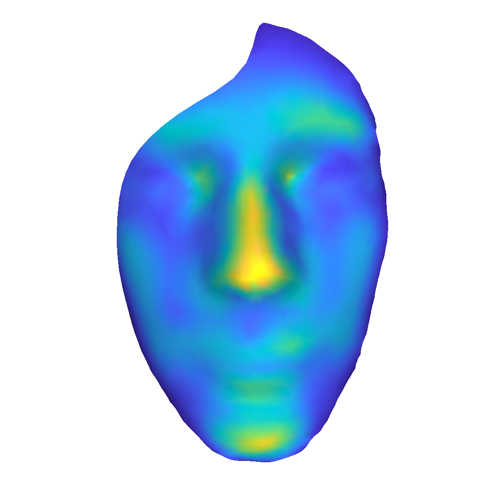}}
     \subfloat{\includegraphics[height=4.5cm]{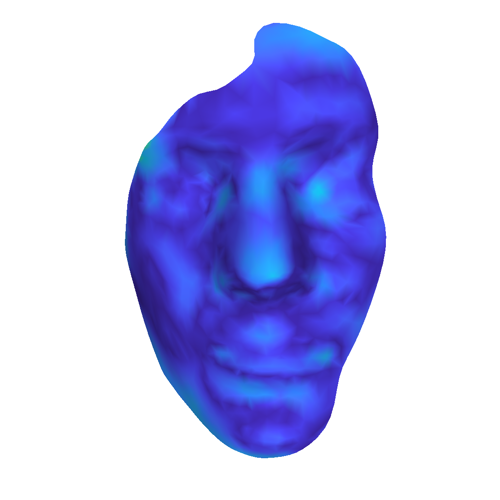}} 
     \subfloat{\includegraphics[width=18pt]{figures/colorbar.png}}
     \caption{Surface registration for a pair of inconsistent human faces in Example 5. The left shows the input moving surface The right shows the target static surface. The right shows the difference of intensities on the registered surface.}
     \label{fig:face_data}
 \end{figure}

 \begin{figure}
     \centering
      \subfloat{\includegraphics[height=5cm]{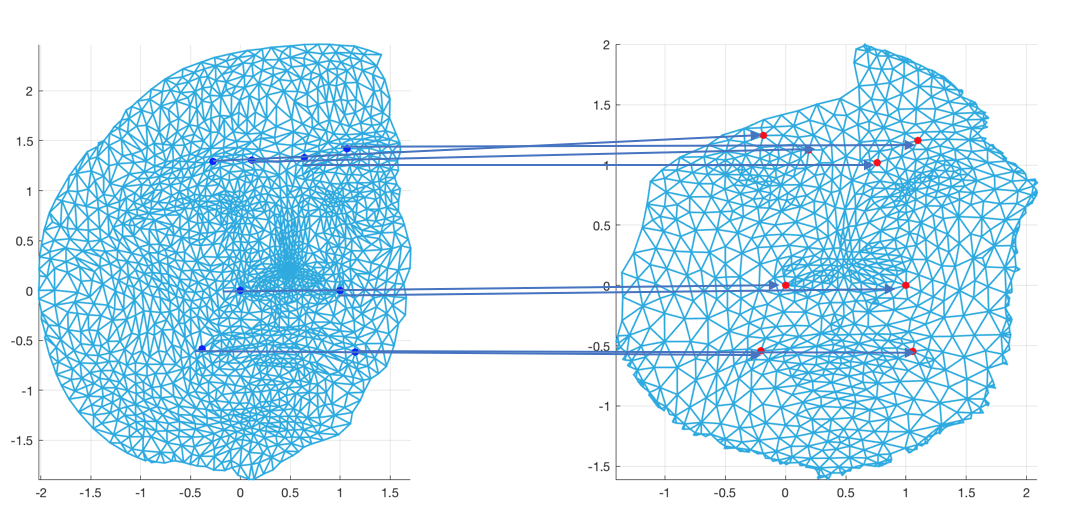}} \\
     \subfloat{\includegraphics[height=5cm]{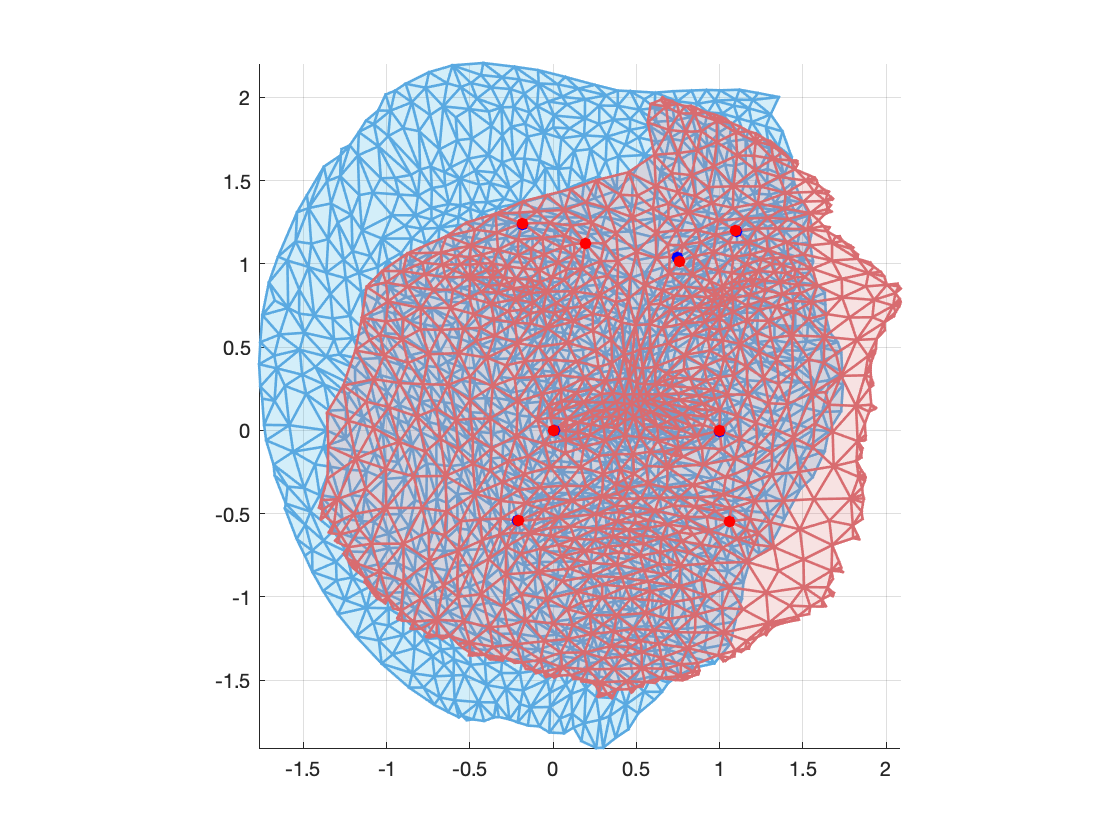}}
     \caption{The first row shows the conformal parametrizations of the moving and static human faces in Example 5 respectively. The landmark correspondences in the 2D parameter domains are also displayed. The registration result in the 2D domain is shown in the second row. The blue mesh is transformed mesh from the moving mesh under the deformation map. The red mesh is the 2D mesh of the target surface under the conformal parametrization.}
     \label{fig:face}
 \end{figure}

\bigskip

\noindent {\bf Example 4}{\it (Registration of partial tooth surfaces)}: Here, we test our registration model on another inconsistent pair of teeth. In Figure \ref{fig:teeth2_data}, the first and second columns show the moving tooth surface and the target static tooth surface respectively. The two surfaces are obviously not bijectively corresponding to each other. Again, our goal is to simultaneously look for the corresponding regions on each surface as well as the registration map between them. The colormaps on each surface are given by their curvatures, which are to be matched using our registration algorithm. As before we first perform conformal flattening \cite{levy2002least} of the two surface meshes into 2D, as shown in the first row in Figure \ref{fig:teeth2}. The landmark correspondences in the 2D parameter domains are also shown. In this example, we have used $\alpha=0.06, \beta = 0.11$, smoothing steps $M_1=1, M_2=5$; bounds $K_1 = 1.3, K_2 = 0.4$; free boundary subproblem for $N_1 = 1$ iterations; intensity subproblem for $1$ iterations and overall iteration $N=50$. The registration result is shown in the last column of Figure \ref{fig:teeth2_data}. It is the registered surface from the moving surface to the target static surface. The colormap on the surface is given by the curvature mismatching error. Note that the mismatching error is small, indicating the curvatures are accurately matched. The registration result in the 2D domain is shown in the second row of Figure \ref{fig:teeth2}. The blue mesh is transformed mesh from the moving mesh under the deformation map. The red mesh is the 2D mesh of the target surface under the conformal parametrization. The intersection region of the two meshes is the region of correspondence amongst the two tooth surfaces. We also display the energy plot against iteration number averaged per face of the mesh on the left of Figure \ref{fig:teeth2_stat}, and the landmark error plot on the left. We can observe that the algorithm successfully reduces the intensity and landmark matching errors.

\bigskip

\noindent {\bf Example 5} {\it (Registration of a pair of human faces)}: In this example, we test our registration method on a pair of inconsistent human faces, which are obtained from FIDENTIS 3D Face Database \cite{Fidentis}. The first and second columns in Figure \ref{fig:face_data} show the moving human face and the target static human face respectively. The colormaps on each surface are given by their mean curvatures. The two surfaces are not bijectively corresponding to each other. As before, we first perform conformal flattening \cite{levy2002least} of the two surface meshes into 2D, as shown in the first row of Figure \ref{fig:face}. The landmark correspondences in the 2D parameter domains are also shown. In this example, we have used $\alpha=0.1, \beta = 0.1$, smoothing steps $M_1=1, M_2=10$; bounds $K_1 = 1.2, K_2 = 0.8$; free boundary subproblem for $N_1 = 5$ iterations; intensity subproblem for $10$ iterations and overall iteration $N=50$. The registration result is shown in the last column of Figure \ref{fig:face_data}. It is the registered surface from the moving surface to the target static surface. The colormap on the surface is given by the curvature mismatching error. Note that the mismatching error is small, indicating the curvatures are accurately matched. The registration result in the 2D domain is shown in the second row of Figure \ref{fig:face}. The blue mesh is transformed mesh from the moving mesh under the deformation map. The red mesh is the 2D mesh of the target surface under the conformal parametrization. The intersection region of the two meshes is the region of correspondence amongst the two human faces. We also display the energy plot against iteration number averaged per face of the mesh on the left of Figure \ref{fig:face_stat}, and the landmark error plot on the left. The intensity registration is more complicated in this case. Nevertheless we can see that the algorithm still successfully reduces the intensity and landmark matching errors. 
 
%   \begin{figure}[t]
%      \centering
     
%      \subfloat{\includegraphics[width=15pt]{figures/colorbar.png}}\hspace{10pt}
%      \subfloat{\includegraphics[width=4.5cm]{figures/face/region002.png}}
%      \caption{Visualization of matching result. On the left it shows the matching error of intensity; on the right it shows the matched region on the moving domain.}
%      \label{fig:face_visual}
%  \end{figure}
\begin{figure}
     \centering
     \subfloat{\includegraphics[height=4.5cm]{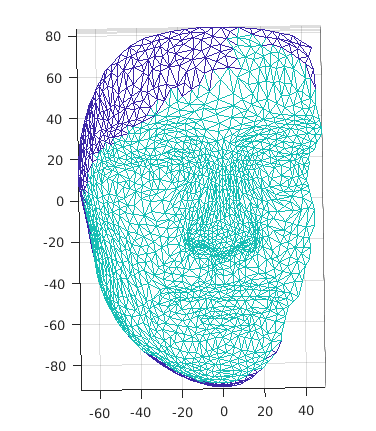}} \hspace{10pt}
     \subfloat{\includegraphics[height=4.5cm]{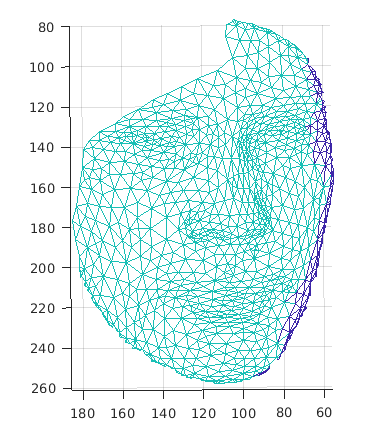}}
     \caption{Corresponding regions on the moving and static human faces in Example 5. The green region on left shows the corresponding region on the moving surface. The green region on the right shows the corresponding region on the target surface.}
     \label{fig:face_corr}
 \end{figure}
 
\begin{figure}
     \centering
     \subfloat{\includegraphics[width=7.2cm]{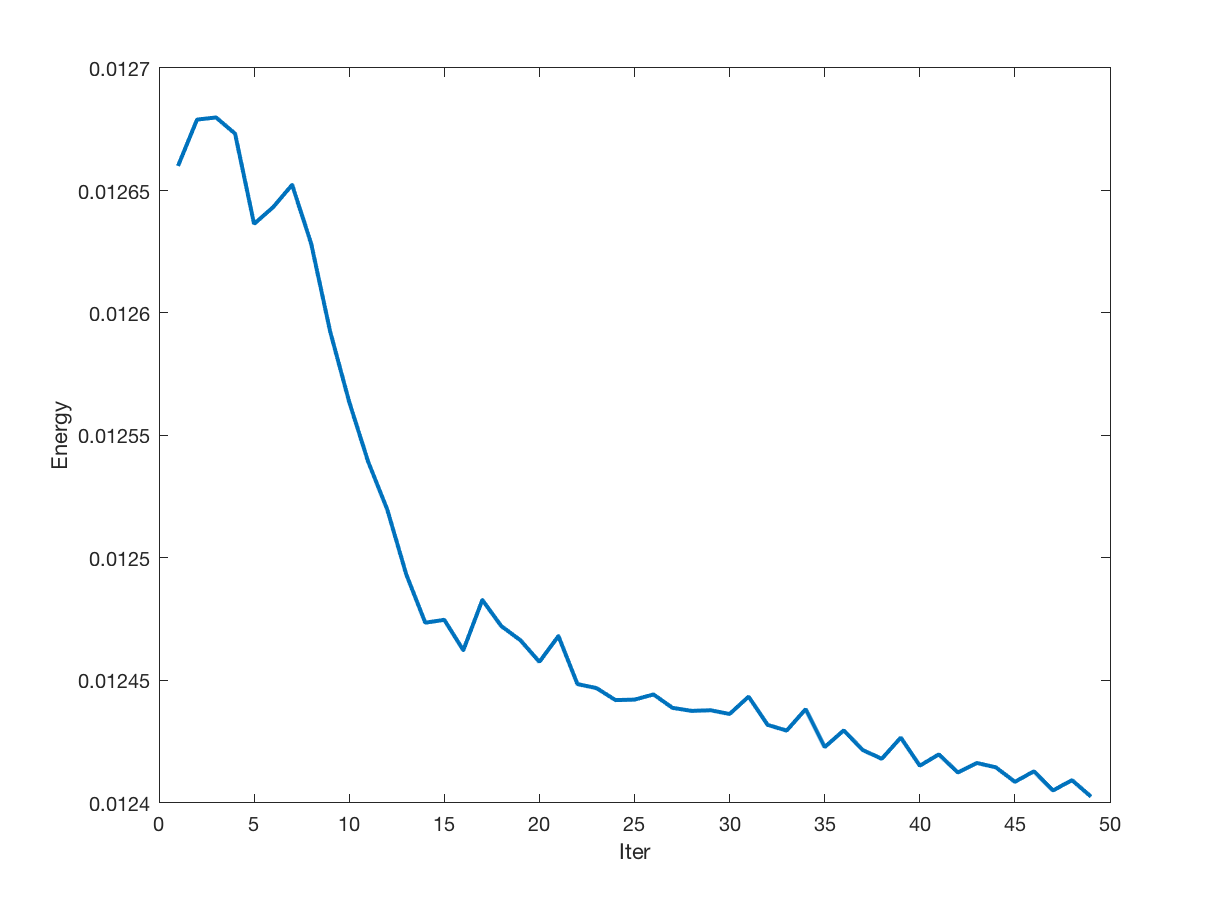}} \hspace{10pt}
     \subfloat{\includegraphics[width=7.2cm]{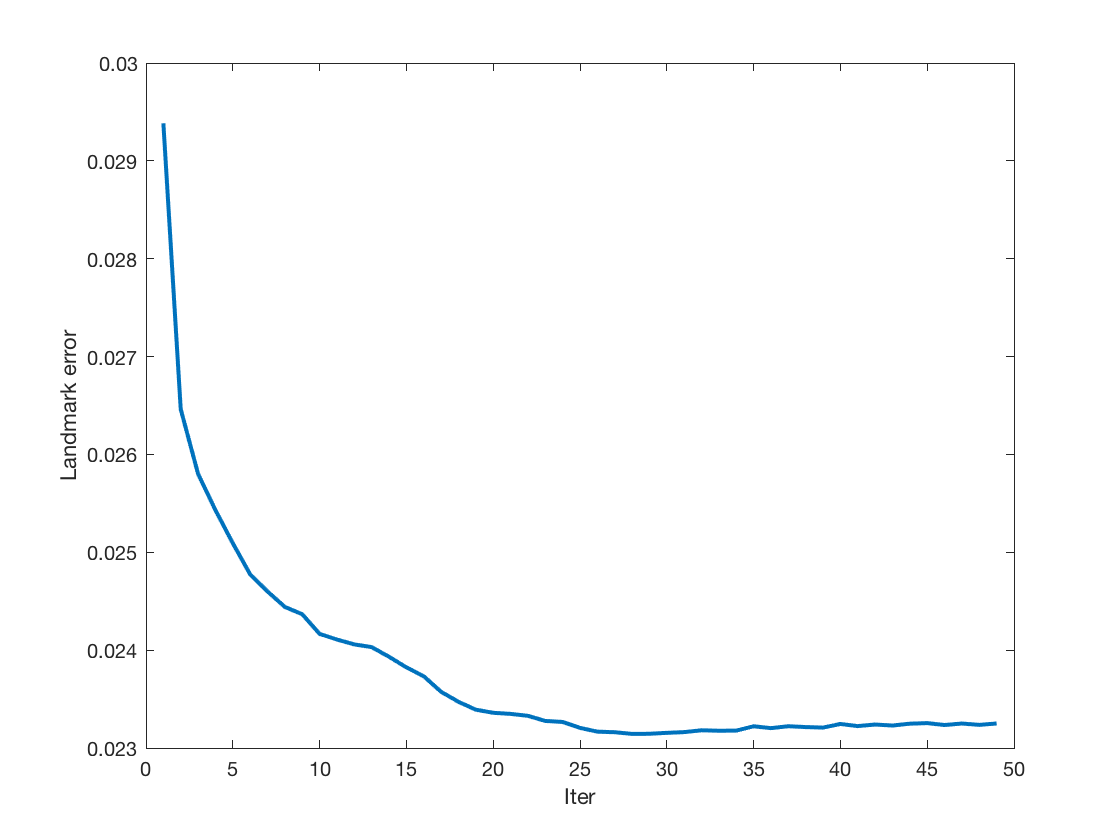}}
     \caption{Energy and landmark error plots for the human face experiment in Example 5 against iteration number. Energy is averaged per face. The left shows the overall energy versus iterations. The right shows the landmark mismatching error versus iterations.}
     \label{fig:face_stat}
 \end{figure}

 \begin{figure}
     \centering
     \subfloat{\includegraphics[height=4.5cm]{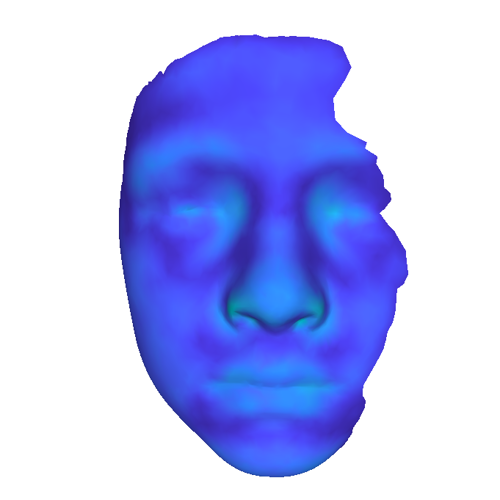}}
     \subfloat{\includegraphics[height=4.5cm]{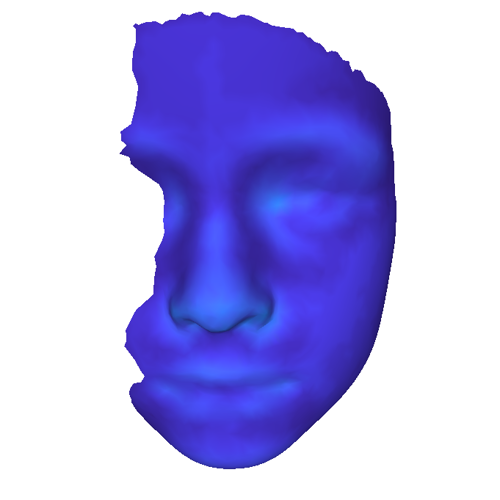}}
     \subfloat{\includegraphics[height=4.5cm]{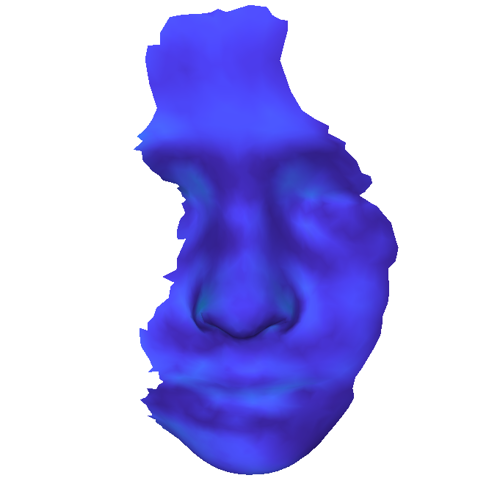}} 
     \subfloat{\includegraphics[width=18pt]{figures/colorbar.png}}
     \caption{Surface registration for another pair of inconsistent human faces in Example 6. The left shows the input moving surface The right shows the target static surface. The right shows the difference of intensities on the registered surface.}
    
     \label{fig:face2_data}
 \end{figure}

 \begin{figure}
     \centering
      \subfloat{\includegraphics[height=5.5cm]{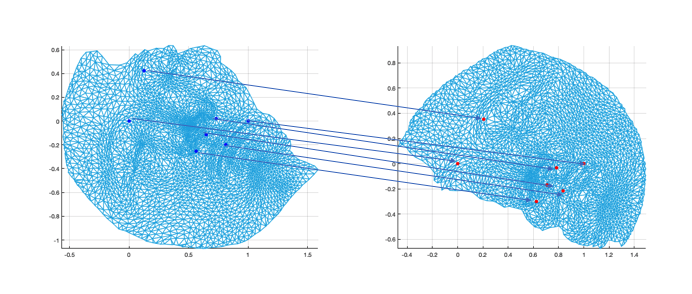}} \\
     \subfloat{\includegraphics[height=5.5cm]{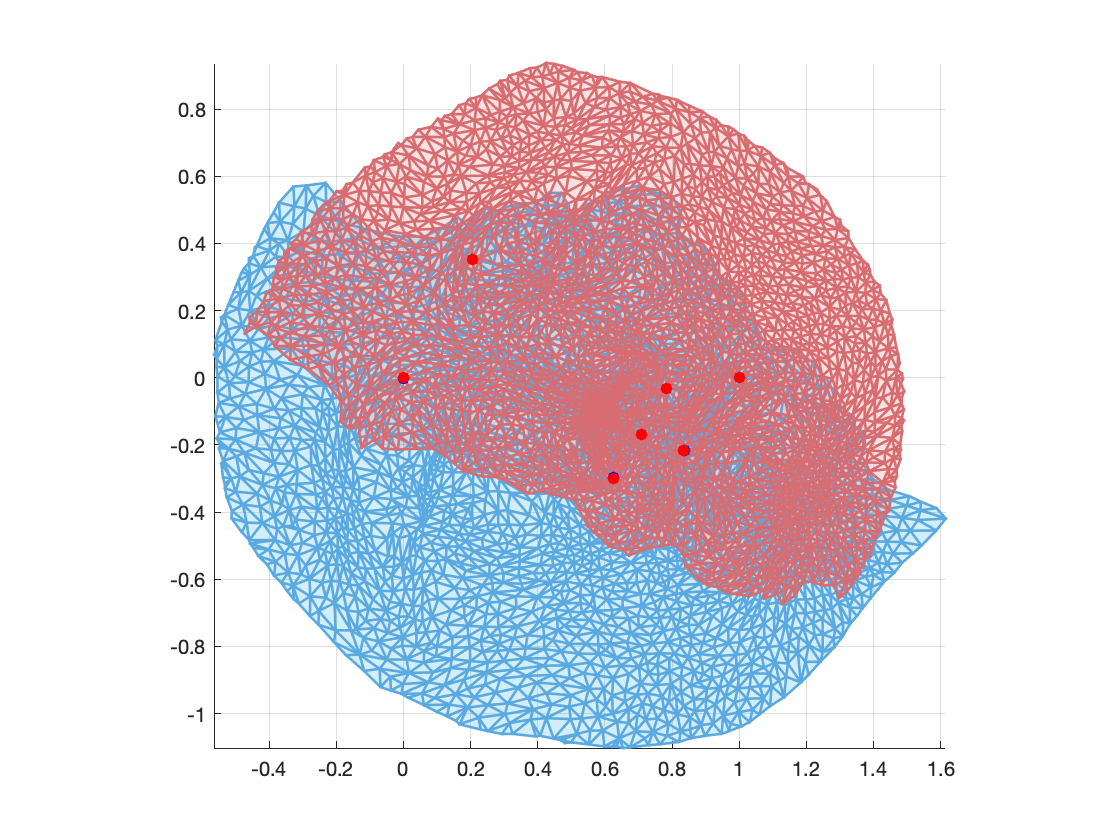}}
     \caption{The first row shows the conformal parametrizations of the moving and static human faces in Example 6 respectively. The landmark correspondences in the 2D parameter domains are also displayed. The registration result in the 2D domain is shown in the second row. The blue mesh is transformed mesh from the moving mesh under the deformation map. The red mesh is the 2D mesh of the target surface under the conformal parametrization.}
     \label{fig:face2}
 \end{figure}
 
 \begin{figure}
     \centering
     \subfloat{\includegraphics[height=4.5cm]{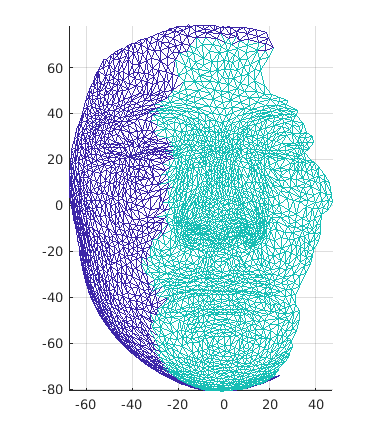}} \hspace{10pt}
     \subfloat{\includegraphics[height=4.5cm]{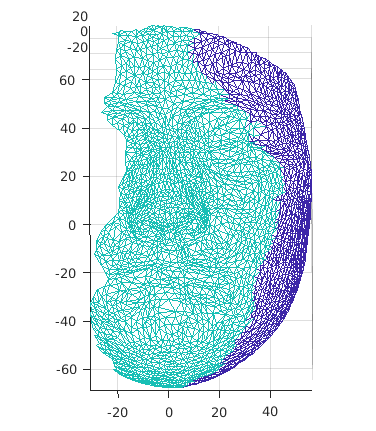}}
     \caption{Corresponding regions on the moving and static human faces in Example 6. The green region on left shows the corresponding region on the moving surface. The green region on the right shows the corresponding region on the target surface.}
     \label{fig:face2_corr}
 \end{figure}

    \begin{figure}
     \centering
     \subfloat{\includegraphics[width=7.2cm]{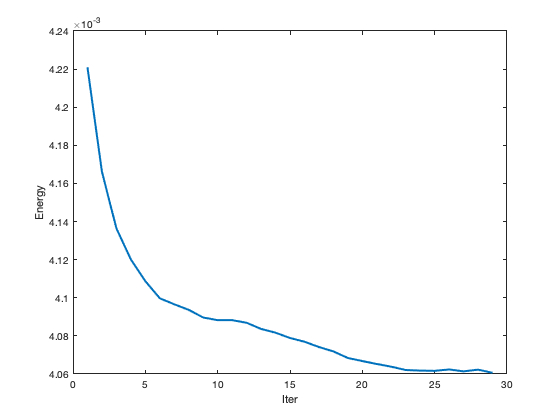}} \hspace{10pt}
     \subfloat{\includegraphics[width=7.2cm]{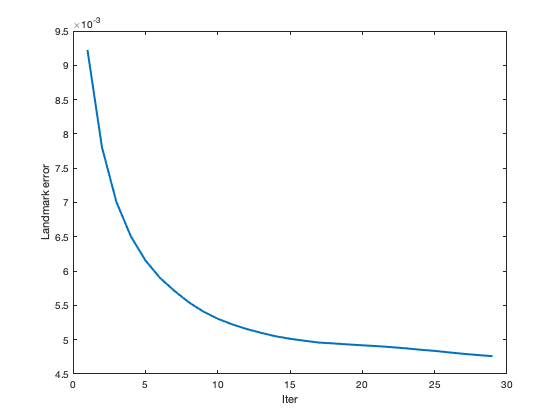}}
     \caption{Energy and landmark error plots for the human face experiment in Example 6 against iteration number. Energy is averaged per face. Energy is averaged per face. The left shows the overall energy versus iterations. The right shows the landmark mismatching error versus iterations.}
     \label{fig:face2_stat}
 \end{figure}
 
\bigskip

\noindent {\bf Example 6:}{\it (Registration of partial human faces)} In this example, we test our registration method on another pair of inconsistent human faces, which are obtained from FIDENTIS 3D Face Database \cite{Fidentis}. The first and second columns in Figure \ref{fig:face2_data} show the moving human face and the target static human face respectively. The matching intensity are shown ont he surfaces. We again perform conformal flattening \cite{levy2002least} of the two surface meshes into 2D, as shown in the first row of Figure \ref{fig:face2}. The landmark correspondences in the 2D parameter domains are also shown. In this example, we have used $\alpha=0.01, \beta = 0.1$, smoothing steps $M_1=1, M_2=5$; bounds $K_1 = 5, K_2 = 0.2$; free boundary subproblem for $N_1 = 5$ iterations; intensity subproblem for $5$ iterations and overall iteration $N=30$. The registration result is shown in the last column of Figure \ref{fig:face2_data}. It is the registered surface from the moving surface to the target static surface. The colormap on the surface is given by the intensity mismatching error. The registration result in the 2D domain is shown in the second row of Figure \ref{fig:face2}. The blue mesh is transformed mesh from the moving mesh under the deformation map. The red mesh is the 2D mesh of the target surface under the conformal parametrization. The intersection region of the two meshes is the region of correspondence amongst the two human faces. We also display the energy plot against iteration number averaged per face of the mesh on the left of Figure \ref{fig:face2_stat}, and the landmark error plot on the left. The intensity registration is more complicated in this case. Nevertheless we can see that the algorithm still successfully reduces the intensity and landmark matching errors. 

\bigskip
\noindent {\bf Example 7:}{\it (Registration of partial genus one vertebrae bones)} In this example, we test our registration method on a pair of inconsistent genus one vertebrae bone surfaces. The first and second columns in Figure \ref{fig:bone_data} show the moving and the target static surfaces respectively. The color-maps on each surface are given by their Gaussian curvatures. We perform conformal flattening and initial registration as in \cite{gu2003global,lui2014geometric} for the two surface meshes so that they share the same fundamental domain in 2D, as shown in the first row of Figure \ref{fig:bone}. The landmark correspondences in the 2D parameter domains are also shown. In this example, we have used $\alpha=0.01, \beta = 0.1$, smoothing steps $M_1=1, M_2=3$; bounds $K_1 = 2, K_2 = 0.5$; free boundary subproblem for $N_1 = 1$ iteration; intensity subproblem for $1$ iteration and overall iteration $N=80$. The registration result is shown in the bottom row of Figure \ref{fig:bone_data}. It is the registered surface from the moving surface to the target static surface. The colormap on the surface is given by the intensity mismatching error. The mismatching error is small, indicating the curvatures are accurately matched. The registration result in the 2D domain is shown in the right column of Figure \ref{fig:bone}. The blue mesh is transformed mesh from the moving mesh under the deformation map. Note that the landmarks are matched almost perfectly. The corresponding regions on both static and moving surfaces are shown in \ref{fig:bone_corr}, where blue color indicates that there is no correspondence. We also display the energy plot against iteration number averaged per face of the mesh on the right of Figure \ref{fig:bone_stat}, and the landmark error plot on the left. We can see that our algorithm successfully reduces the intensity and landmark matching errors in this challenging scenario.

 \begin{figure}
     \centering
     \subfloat{\includegraphics[height=4.5cm]{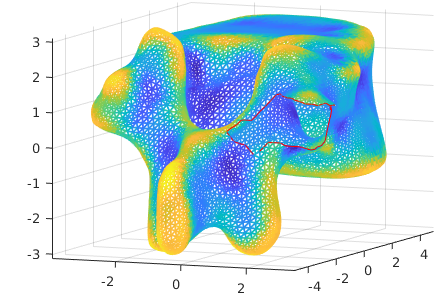}}
     \subfloat{\includegraphics[height=4.5cm]{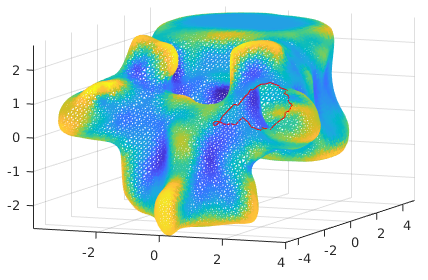}} \\
     \subfloat{\includegraphics[height=5.5cm]{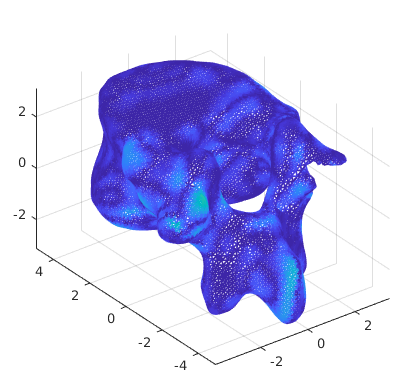}} 
     \subfloat{\includegraphics[width=18pt]{figures/colorbar.png}}
     \caption{Surface registration for a pair of inconsistent vertebrae bone surfaces in Example 7. The left shows the input moving surface. The right shows the target static surface. The colormaps of them are given by their curvatures. The bottom shows the difference of intensities on the registered surface.}
    
     \label{fig:bone_data}
 \end{figure}

 \begin{figure}
     \centering
      \subfloat{\includegraphics[height=5.5cm]{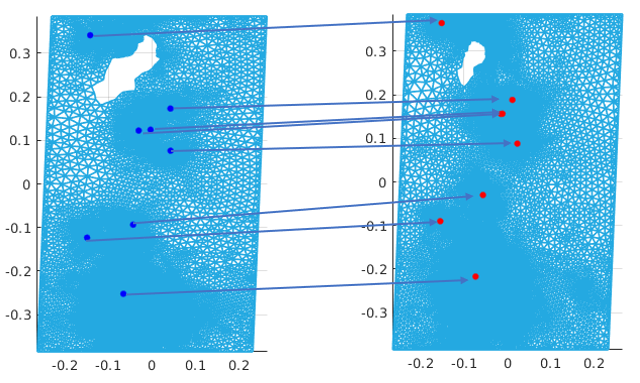}}
     \subfloat{\includegraphics[height=5.5cm]{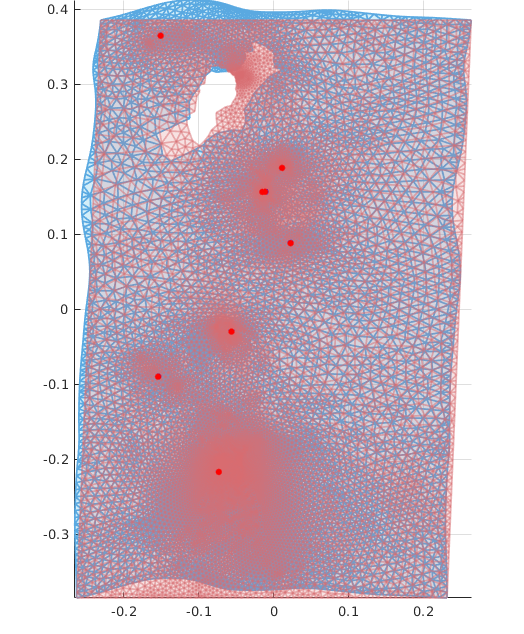}}
    %  \subfloat{\includegraphics[height=5.5cm]{figures/bone/flat_result.png}}
     \caption{One the left we show the conformal parametrizations of the moving and static vertebrae bone surfaces in Example 7. The landmark correspondences in the 2D parameter domains are also displayed. The overlaid view of the registration result in the 2D domain is shown on the right. The blue mesh is transformed mesh from the moving mesh under the deformation map. The red mesh is the 2D mesh of the target surface under the conformal parametrization.}
     \label{fig:bone}
 \end{figure}
 
 \begin{figure}
     \centering
     \subfloat{\includegraphics[height=4.5cm]{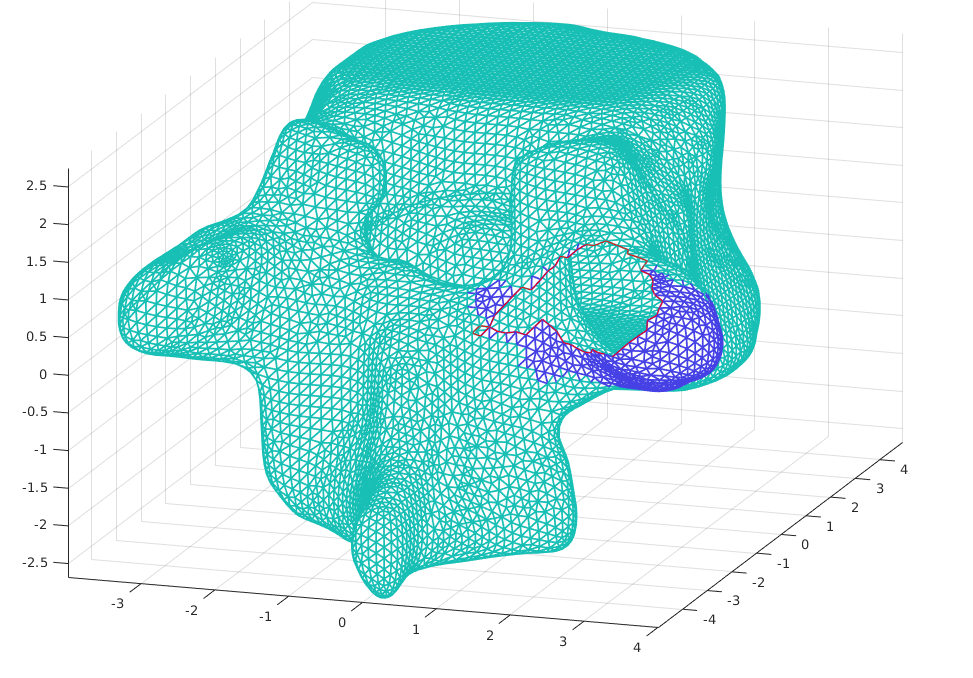}} \hspace{10pt}
     \subfloat{\includegraphics[height=4.5cm]{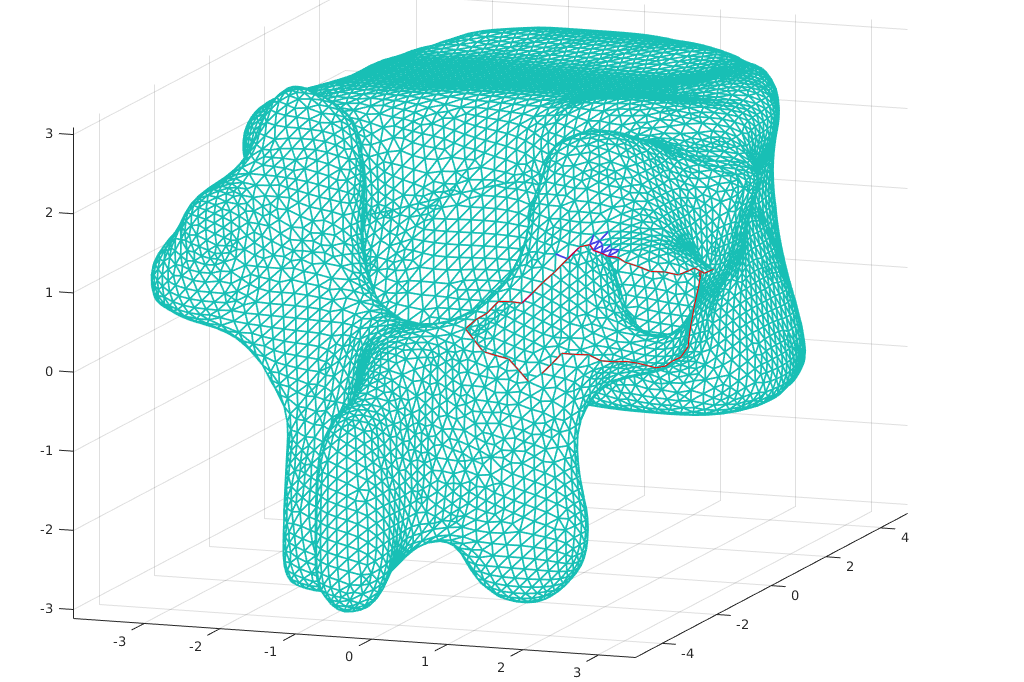}}
     \caption{Corresponding regions on the moving and static human faces in Example 7. The green region on left shows the corresponding region on the moving surface. The green region on the right shows the corresponding region on the target surface.}
     \label{fig:bone_corr}
 \end{figure}

    \begin{figure}
     \centering
     \subfloat{\includegraphics[height=6.5cm]{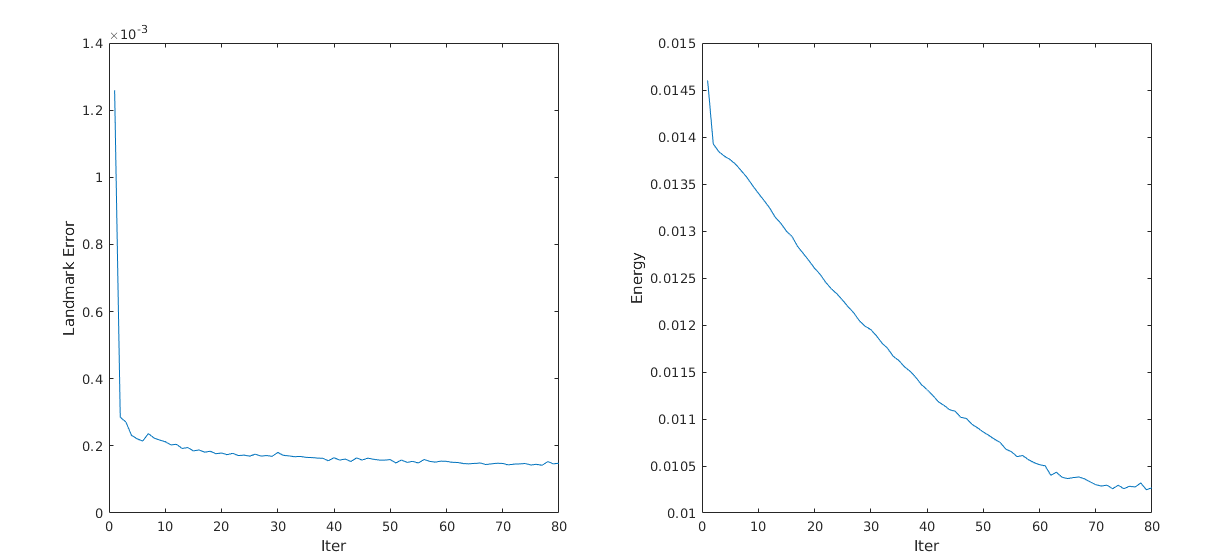}} 
     \caption{Energy and landmark error plots for the vertebrae bone experiment in Example 7. Energy is averaged per face. The left shows the overall energy versus iterations. The right shows the landmark mismatching error versus iterations.}
     \label{fig:bone_stat}
 \end{figure}

 \begin{figure}
     \centering
     \subfloat{\includegraphics[height=4.5cm]{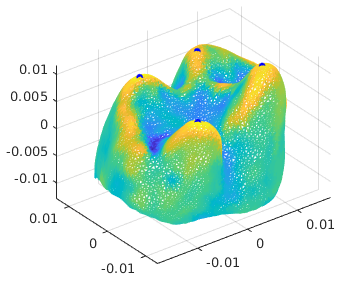}}
     \hspace{-2pt}
     \subfloat{\includegraphics[height=4.5cm]{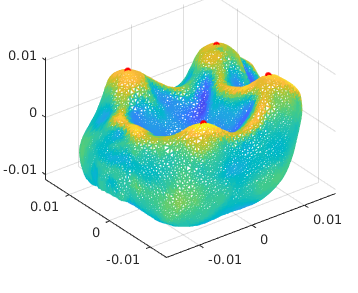}}
     \hspace{-2pt}
     \subfloat{\includegraphics[height=4.5cm]{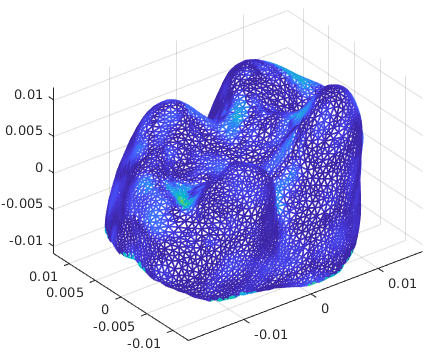}} 
     \subfloat{\includegraphics[width=18pt]{figures/colorbar.png}}
     \caption{Input moving and static  mammalian tooth surfaces for comparison study in Example 8. In the first row the left shows the input moving surface and the right shows the target static surface. The second row shows the difference of intensities on the corresponding region of the moving surface.}
    
     \label{fig:compare_data}
 \end{figure}

 \begin{figure}
     \centering
     \subfloat{\includegraphics[height=4.5cm]{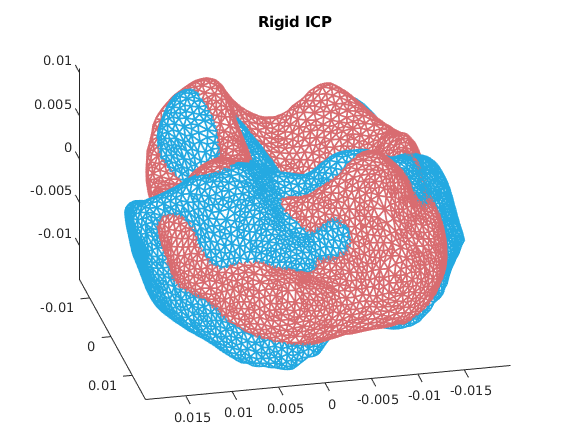}}
     \subfloat{\includegraphics[height=4.5cm]{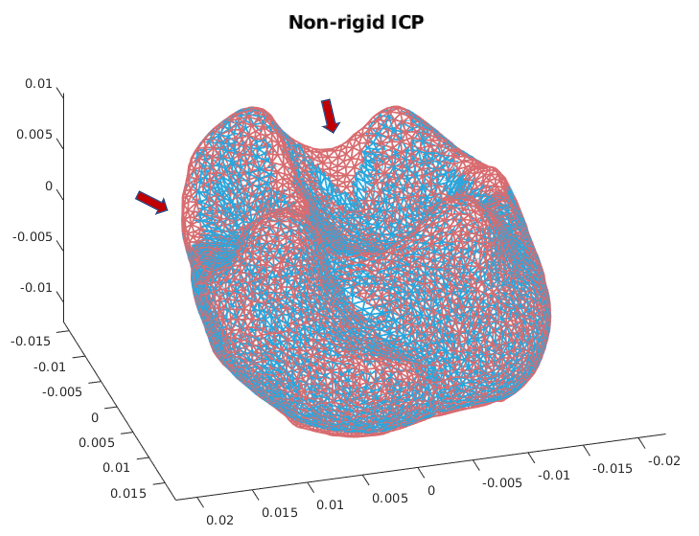}} \\
     \subfloat{\includegraphics[height=4.5cm]{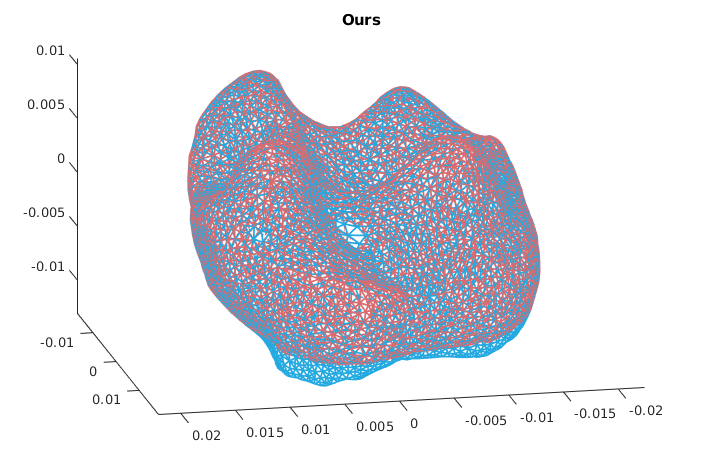}}
     \caption{Comparison on a pair of mammalian tooth surfaces in Example 8. The blue mesh is transformed under the deformation map to match the red mesh. Rigid ICP method can only find a rotation and translation, and hence no correspondence is obtained between the two surfaces. Non-rigid ICP method tries to incorporate deformable motion, but fails to produce correspondence in many places, as indicated by the red arrows. In contrast, our registration method successfully find a diffeomorphic mapping between sub-regions of the moving and static surfaces that minimizes both intensity and landmark mismatch.}
     \label{fig:compare}
 \end{figure}
\bigskip

\noindent {\bf Example 8:}{\it (Comparison with rigid and non-rigid iterative closest point methods)} In this example, we compare our registration method with both rigid and non-rigid iterative closest point (ICP) methods on a pair of mammalian tooth surfaces with initial alignment. The first rows in Figure \ref{fig:compare_data} show the moving and the target static surfaces respectively. The color-maps on each surface are given by their Gaussian curvatures, where the landmark correspondences are also shown. The second row shows the intensity difference after registration using our approach. We can see the intensities are accurately matched.  In this example, we have used $\alpha=0.01, \beta = 0.1$, smoothing steps $M_1=1, M_2=5$; bounds $K_1 = 2, K_2 = 0.5$; free boundary subproblem for $N_1 = 1$ iteration; intensity subproblem for $1$ iteration and overall iteration $N=20$. The registration results of the all three methods are shown in Figure \ref{fig:compare}, where static surface is in red and moving surface is in blue. Conventional rigid ICP method failed because the underlying registration is deformable and more complex than merely rotation and translation. Non-rigid ICP method tries to incorporate deformable motion, but fails to produce correspondence in many places, as indicated by the red arrows. In contrast, our registration method successfully find a diffeomorphic mapping between sub-regions of the moving and static surfaces that minimizes both intensity and landmark mismatch.

\section{Conclusion} \label{sec:7}
We have proposed a deformation model that is able to control both angle and area distortion and allow free boundary movement, which is further developed into a registration algorithm used for domains that are not in any natural global bijective correspondence. The key is to use Beltrami coefficient for smoothness and the mapping differential singular values for free boundary deformation. Experimental results have been given to show the effectiveness of our approach.

\clearpage 

\bibliographystyle{spmpsci} 
\bibliography{reference}
\end{document}